\documentclass[sigconf]{acmart}
\AtBeginDocument{%
  }

\usepackage{tcolorbox}
\usepackage{graphicx}
\usepackage{subcaption}
\usepackage{multirow}
\usepackage{amsthm}
\newtheorem{theorem}{Theorem}[section]
\usepackage{algorithm}
\usepackage{algpseudocode}
\usepackage{xcolor}
\usepackage{algpseudocode}
\usepackage[absolute,overlay]{textpos}
\usepackage{amsmath}

\setcopyright{acmlicensed}
\copyrightyear{2018}
\acmYear{2018}
\acmDOI{XXXXXXX.XXXXXXX}
\acmConference[Conference acronym 'XX]{Make sure to enter the correct
  conference title from your rights confirmation email}{June 03--05,
  2018}{Woodstock, NY}
\acmISBN{978-1-4503-XXXX-X/2018/06}

\begin{document}

\begin{textblock*}{\paperwidth}(2.0cm, 1.5cm)\includegraphics[height=1.3cm]{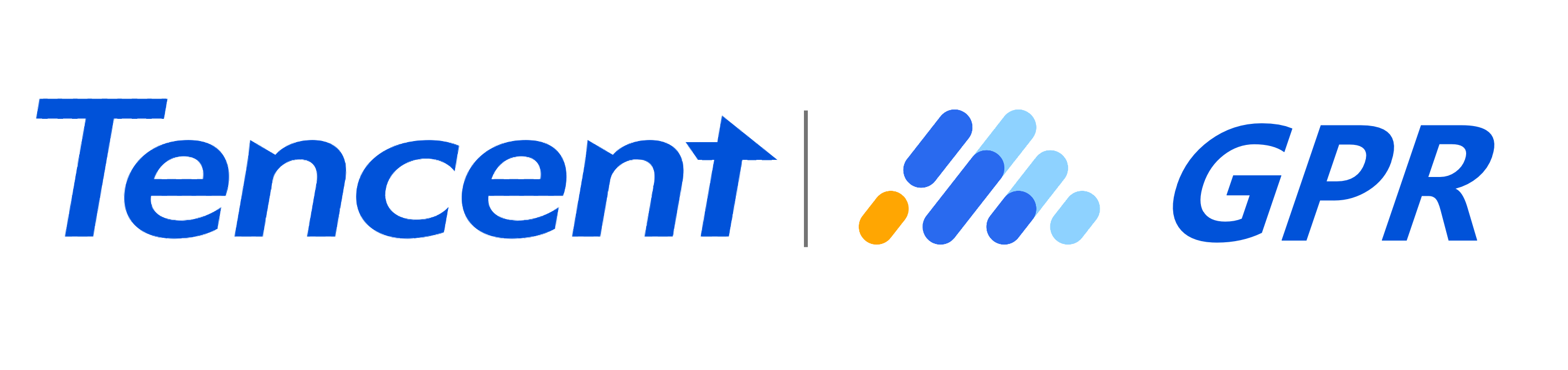} \end{textblock*}

\title{SCoTER: Structured Chain-of-Thought Transfer for Enhanced Recommendation}


\author{Jie Jiang}
\authornote{Both authors contributed equally to this research.} 
\affiliation{%
  \institution{Tencent}
    \state{Beijing}
  \country{China}
}
\email{zeus@tencent.com}

\author{Yang Wu}
\authornotemark[1]
\affiliation{%
  \institution{Tencent}
  \state{Beijing}
  \country{China}
}
\email{samuelywu@tencent.com}

\author{Qian Li}
\authornotemark[1]
\affiliation{%
  \institution{Tencent}
  \state{Beijing}
  \country{China}
}
\email{kathieqli@tencent.com}

\author{Yuling Xiong}
\affiliation{%
  \institution{Tencent}
  \state{Beijing}
  \country{China}
}
\email{whitnyxiong@tencent.com}

\author{Hongbo Tang}
\affiliation{%
  \institution{Tencent}
    \state{Beijing}
  \country{China}
}
\email{hardytang@tencent.com}

\author{Xun Liu}
\affiliation{%
  \institution{Tencent}
    \state{Beijing}
  \country{China}
}
\email{	reubenliu@tencent.com}

\author{Haoze Wang}
\affiliation{%
  \institution{Tencent}
    \state{Beijing}
  \country{China}
}
\email{howsonwang@tencent.com}

\author{Jun Zhang}
\authornote{Corresponding author.}
\affiliation{%
  \institution{Tencent}
    \state{Beijing}
  \country{China}
}
\email{neoxzhang@tencent.com}

\author{Huan Yu}
\affiliation{%
  \institution{Tencent}
    \state{Beijing}
  \country{China}
}
\email{	huanyu@tencent.com}

\author{Hailong Shi}
\authornotemark[2]
\affiliation{%
  \institution{Chinese Academy of Sciences}
    \state{Beijing}
  \country{China}
}
\email{shihailong2010@gmail.com}

\renewcommand{\shortauthors}{Qian Li}

\begin{abstract}
Harnessing the reasoning power of Large Language Models (LLMs) for recommender systems is hindered by two fundamental challenges. First, current approaches lack a mechanism for automated, data-driven discovery of effective reasoning patterns, relying instead on brittle manual templates or unstable zero-shot prompting. Second, they employ structure-collapsing integration: direct prompting incurs prohibitive online inference costs, while feature extraction collapses reasoning chains into single vectors, discarding stepwise logic. To address these challenges, we propose SCoTER (Structured Chain-of-Thought Transfer for Enhanced Recommendation), a unified framework that treats pattern discovery and structure-aware transfer as a jointly optimized problem. Specifically, SCoTER operationalizes this through two synergistic components: a Generate-Validate-Mine (GVM) pipeline for automated pattern discovery and a structure-preserving integration architecture that transfers stepwise logic to efficient models. Empirically, experiments on four benchmarks demonstrate consistent improvements across diverse backbones. Moreover, in production deployment on the Tencent Advertising Platform, SCoTER achieved a 2.14\% lift in Gross Merchandise Value (GMV) while eliminating online LLM inference costs. Overall, SCoTER presents a practical and unified framework for integrating structured LLM reasoning into recommender systems, validated by consistent improvements in both offline benchmarks and online production environments.
\end{abstract}





\maketitle
\section{Introduction}

Large Language Models (LLMs)~\cite{llmsurvey1, llmsurvey2} have demonstrated strong reasoning capabilities, especially with Chain-of-Thought (CoT) prompting~\cite{wei2022chain}. Recent studies~\cite{tsai2024leveraging, liu2025improving, bismay2025reasoningrec, zhao2025reason, you2025text} have therefore explored applying this reasoning power to recommender systems. Unlike conventional approaches that predict next items without explicit intermediate reasoning, such methods utilize step-by-step reasoning to infer the underlying rationale connecting historical behaviors to potential interests. This logical modeling capability not only improves recommendation accuracy but also provides more interpretable explanations. Consequently, integrating such structured reasoning has emerged as a promising direction for advancing modern recommender systems~\cite{wu2024survey, zhao2024recommender}.

\begin{figure}[t]
  \centering
  \includegraphics[width=\linewidth]{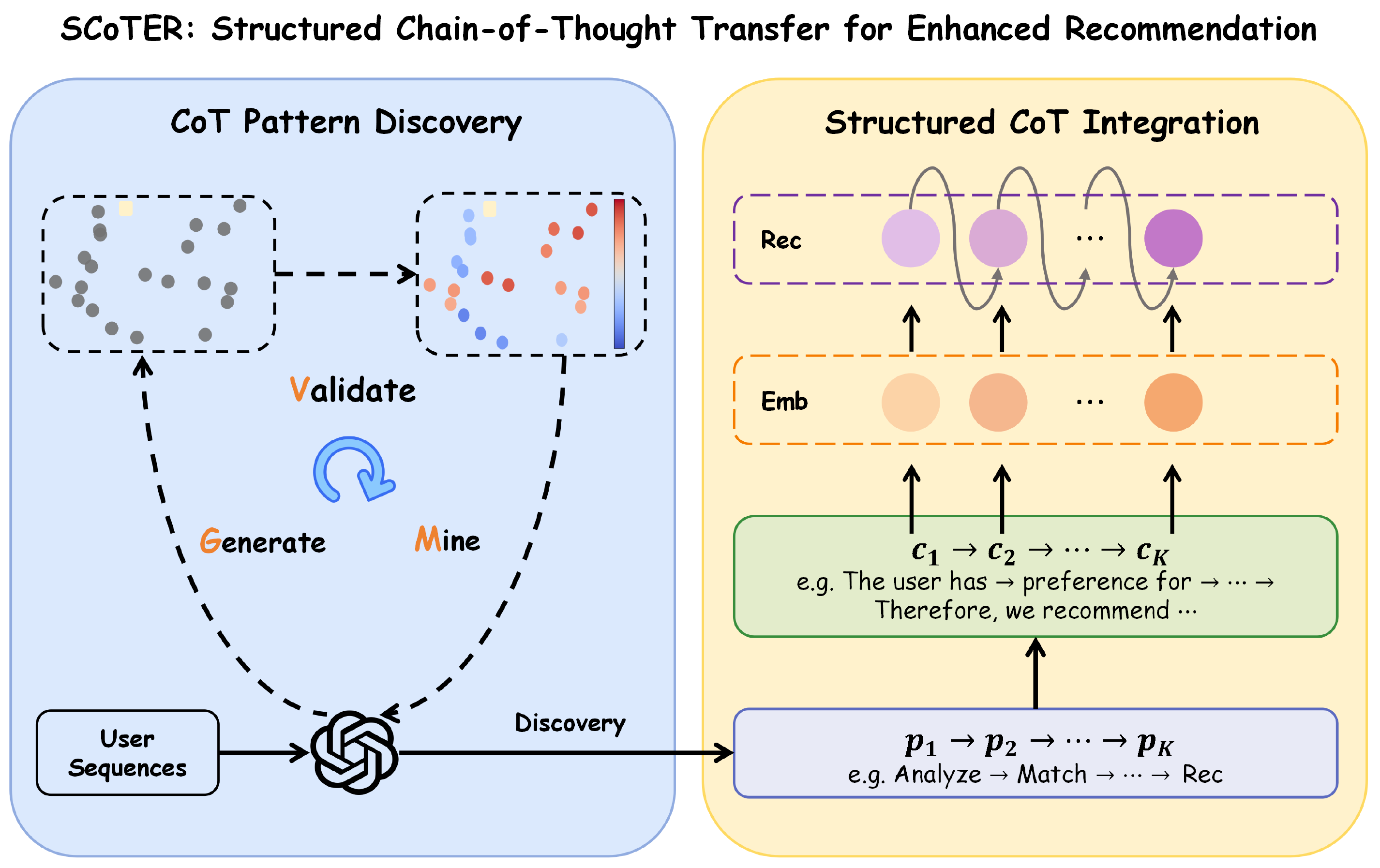}
  \caption{An overview of the SCoTER framework. It consists of two main components. First, the CoT Pattern Discovery pipeline addresses what to transfer by automatically discovering effective reasoning patterns from data. Second, the Structured CoT Integration architecture addresses how to transfer by integrating these patterns into an efficient model while preserving their step-wise structure.}
  \label{fig:intro-contrast}
\end{figure}

Despite this progress, existing approaches largely rely on manually crafted templates or heuristic rules to generate reasoning traces. In contrast to objective tasks with canonical reasoning structures, recommendation typically lacks ground-truth reasoning paths, as user intent is often latent, context-dependent, and highly diverse. As a result, static handcrafted designs can be brittle~\cite{wang2023self, zhangautomatic, wang2023plan}, failing to generalize across heterogeneous behaviors or to capture implicit interaction patterns. This exposes a fundamental “what to transfer” challenge: how to automatically discover effective, data-driven reasoning patterns that align with user intent, rather than relying on hand-engineered heuristics.

Existing work has explored incorporating LLM reasoning into recommender systems primarily in two ways. One line of methods directly prompts LLMs to generate step-by-step reasoning at inference time~\cite{bao2023tallrec, liao2024llara, zheng2024adapting, xia2025trackrec}, preserving the sequential logic of CoT but incurring prohibitive latency and computational costs for online deployment. Another line leverages offline feature extraction~\cite{wang2024can, cui2024distillation, xi2024towards, yang2025reccot}, using LLM-generated reasoning traces to help conventional recommenders for efficient inference. However, many such pipelines adopt a “structure-collapsing” strategy, compressing multi-step reasoning into a single fixed representation~\cite{wang2024can, xi2024towards, yang2025reccot, xia2025trackrec}. This compression weakens the stepwise dependencies that underlie the inference process, reducing rich reasoning chains to static signals. We refer to this as the “how to transfer” challenge. Furthermore, addressing what to transfer and how to transfer in isolation prevents joint optimization—reasoning patterns may be discovered without regard for integration costs, while transfer mechanisms may fail to preserve core stepwise logic.


To address these intertwined challenges, we introduce SCoTER (Structured Chain-of-Thought Transfer for Enhanced Recommendation), a unified framework designed to jointly optimize pattern discovery and structure-preserving integration (Figure~\ref{fig:intro-contrast}). To address the ‘what to transfer’ challenge, SCoTER features the Generate-Validate-Mine (GVM) pipeline. This pipeline transforms pattern discovery from a heuristic exercise into a data-driven optimization process: an LLM first generates a diverse set of candidate reasoning paths, which are then validated based on the empirical quality of their recommendations, before a final mining process distills the most effective and generalizable pattern. For the ‘how to transfer’ challenge, SCoTER employs a lightweight, structure-preserving architecture. This component integrates pre-computed offline reasoning embeddings via an order-aware fusion mechanism, which preserves the sequential structure of CoT while eliminating prohibitive online LLM inference costs. Extensive experiments demonstrate consistent improvements across diverse backbones on four public benchmarks, and a large-scale production deployment further yields a 2.14\% GMV lift.

Our main contributions are:
\begin{itemize}
  \item \textbf{Automated Discovery Pipeline:} We introduce the GVM pipeline to automatically discover data-driven reasoning patterns, reducing reliance on manual templates.
  
  \item \textbf{Structure-Preserving Integration:} We propose a lightweight architecture using pre-computed stepwise embeddings and order-aware fusion, eliminating online LLM inference while preserving sequential dependencies.
  
  \item \textbf{Comprehensive Validation:} Experiments demonstrate consistent improvements across four benchmarks and a 2.14\% GMV lift in production.
\end{itemize}
\section{Related works}

\textbf{LLM Reasoning for Recommendation.} Recent approaches integrating LLM into recommendation systems have explored a variety of complex reasoning structures. CoT-Rec \cite{liu2025improving} employs two-stage prompting for user preference analysis, GOT4Rec \cite{long2024got4rec} uses Graph-of-Thought frameworks, and ThinkRec \cite{yu2025thinkrec} shifts to System 2 thinking through reasoning data synthesis, while RecGPT \cite{yi2025recgpt} works to unify multi-step reasoning frameworks. Complementing these efforts, a parallel line of research focuses on refining or distilling reasoning capabilities. This includes the inference-time autoregressive refinement in ReaRec \cite{tang2025think}, the distillation of step-by-step rationales to smaller models by RDRec \cite{rdrec}, and the iterative feedback framework used by TrackRec \cite{xia2025trackrec}. However, these methods are limited by relying on heuristic reasoning paths instead of mining user sequences, and their failure to jointly optimize pattern discovery and integration.

\textbf{Automated Reasoning Discovery}: Automated discovery of reasoning patterns has emerged as an alternative to manual template design. Auto-CoT \cite{zhangautomatic} automatically constructs demonstrations by sampling diverse questions and generating rationales, while Self-prompted CoT \cite{wang2023self} enables LLMs to self-induce reasoning steps. Self-Consistency \cite{wangself} improves reasoning by sampling multiple paths, and broader approaches include APE \cite{zhou2022large} for automatic prompt engineering, PromptBreeder \cite{fernando2024promptbreeder} for evolutionary optimization, and Self-discover \cite{zhou2024self} for composing atomic reasoning modules.  However, these methods are primarily designed for objective tasks with verifiable ground truth. They are less effective in the recommendation domain, which is subjective and has sparse rewards that make it difficult to improve reasoning paths. Our approach, in contrast, addresses this by sampling from broad user behaviors and performing in-depth analysis to use the Recall metric as a dense reward signal.

\begin{figure*}[t]
    \includegraphics[width=1\linewidth]{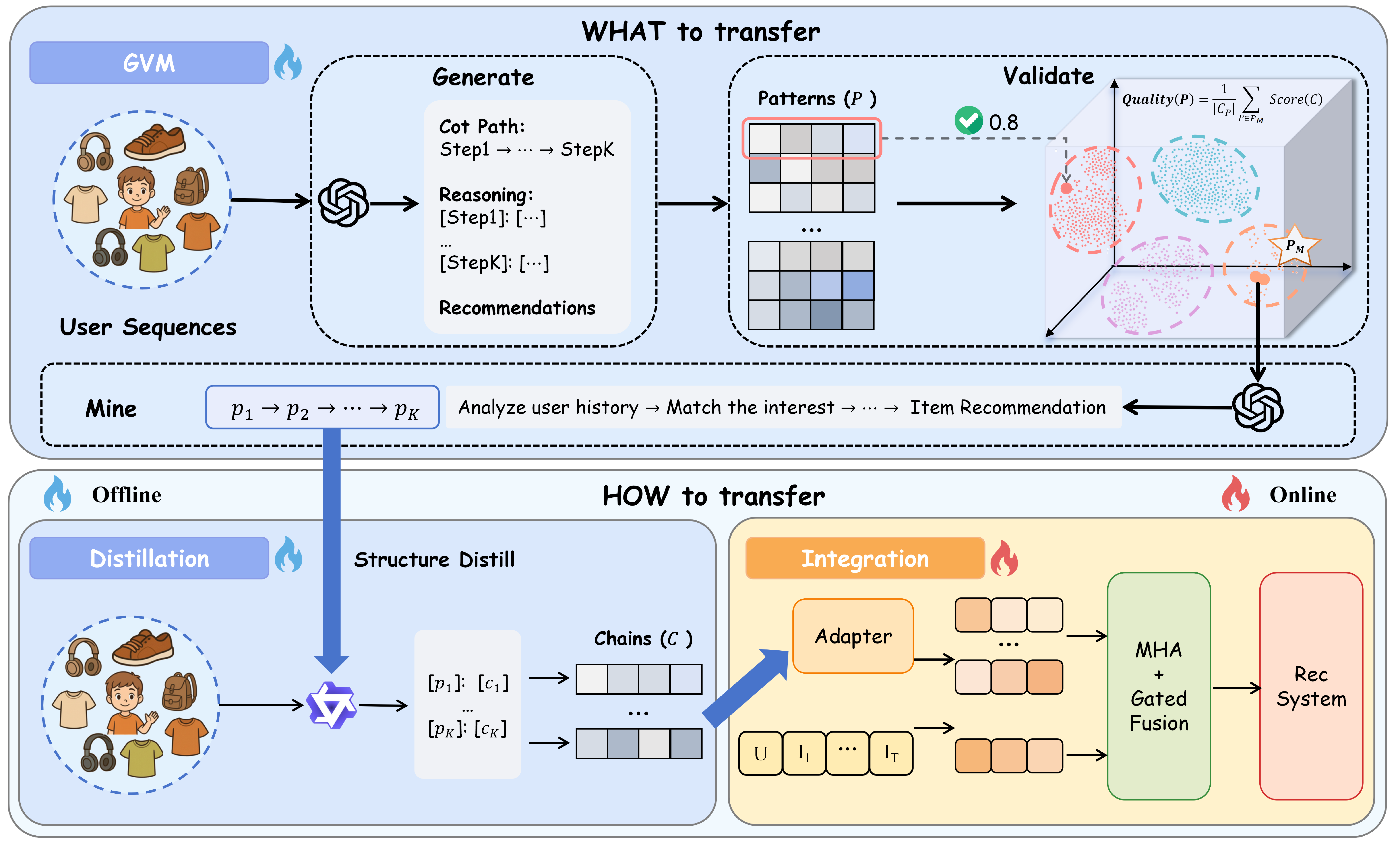}
    \caption{The proposed Structured Chain-of-Thought Transfer for Enhanced Recommendation (SCoTER) framework. It jointly solves the challenges of what to transfer and how to transfer CoT reasoning. To determine what to transfer, an offline GVM (Generate-Validate-Mine) pipeline automates the discovery of optimal reasoning patterns from data. To address how to transfer these patterns, they are first materialized as step-wise embeddings via Structured Distillation and subsequently fused into a backbone model by a lightweight Online Integration module that preserves the chain's logical structure.}
    \label{fig:framework}
\end{figure*}

\section{Method}\label{sec:method}

In this section, we propose SCoTER, a unified framework designed to jointly address the dual challenges of \textit{what to transfer} and \textit{how to transfer}. As shown in Figure~\ref{fig:framework}, the framework consists of two main components: the Generate-Validate-Mine (GVM) pipeline for automated pattern discovery, and a structure-preserving architecture for efficient integration. The full pipeline is formalized in Appendix~\ref{app:pseudocode} (Algorithm~\ref{alg:gvm}--\ref{alg:scoter}). We begin by introducing the problem definition and preliminaries in Section~\ref{subsec:preliminaries}, followed by the details of the two components. 

\subsection{Preliminaries}
\label{subsec:preliminaries}

We begin by formally defining the sequential recommendation task and the hierarchical reasoning structures.

\noindent\textbf{Sequential Recommendation.} Let $\mathcal{U}$ and $\mathcal{I}$ denote the sets of users and items, respectively. For each user $u \in \mathcal{U}$, the historical interaction sequence is defined as $S = [i_1, i_2, \ldots, i_T]$, where $i_t \in \mathcal{I}$ represents the item interacted with at time step $t$. The goal of sequential recommendation is to predict the next item $i_{T+1}$ that the user is most likely to interact with. Mathematically, the recommender aims to learn a function with parameters $\Theta$ that maximizes the prediction probability $p(i_{T+1} | S; \Theta)$.

\noindent\textbf{Reasoning Pattern ($P$).} 
A \textit{Reasoning Pattern} is an abstract template of $K$ logical steps that guides the reasoning process. Formally, $P = (p_1, p_2, \ldots, p_K)$, where each $p_k$ represents a high-level reasoning instruction (e.g., ``Analyze user history'' $\rightarrow$ ``Match the interest'' $\rightarrow$ $\cdots$ $\rightarrow$ ``Item Recommendation'').

\noindent\textbf{Reasoning Chain ($C$).} This is the specific textual instantiation of $P$ generated by an LLM for a sequence $S$. For example, given history $S = [\text{PS5 Console}, \text{FIFA 24}]$, the chain unfolds as:
\begin{itemize}
    \item \textbf{Analyze user history ($c_1$):} ``$\cdots$ User has a new console and likes sports games.$\cdots$ ''
    \item \textbf{Match the interest ($c_2$):} ``$\cdots$ Likely needs local multiplayer accessories. $\cdots$'' 
    \item \textbf{Item Recommendation ($c_K$):} ``$\cdots$ Recommend DualSense Controller.''
\end{itemize}

\noindent\textbf{Vectorization.} To incorporate this textual reasoning into the recommendation model, we employ a pre-trained sentence encoder to map each discrete step $c_k$ into a dense vector $\mathbf{e}_k \in \mathbb{R}^D$. The sequence of vectors is then stacked to form the structured reasoning matrix $\mathbf{H} = [\mathbf{e}_1; \mathbf{e}_2; \dots; \mathbf{e}_K]$, which serves as the reasoning-enhanced input for our fusion module.

\subsection{What to Transfer: Automated Discovery of Reasoning patterns}\label{subsec:off1}

Our approach replaces manual template design with a three-phase optimization pipeline: GVM. This process systematically mines the optimal pattern from a diverse set of candidate reasoning chains and extracts it as a symbolic template for the subsequent transfer.

\textbf{Generate:}
The Generate phase produces a diverse set of candidate reasoning chains for each user sequence $S$. We employ an LLM (e.g., DeepSeek-R1~\cite{guo2025deepseek}) with a structured prompt (Figure~\ref{fig:prompt_structure}) that instructs the model to act as a recommendation expert. The prompt requires three distinct outputs: (1) a concise, step-wise reasoning chain in \texttt{<cot\_path>} tags; (2) a detailed elaboration in a \texttt{<reason>} block; and (3) a list of 20 ranked recommendations in \texttt{<recommendations>} tags. This explicit separation decouples the abstract reasoning logic from detailed explanations.

To ensure diversity, we employ two mechanisms. First, during generation, we use temperature $\tau_{gen}=0.7$ and top-$p$ nucleus sampling with $p=0.9$ to encourage varied reasoning styles. Second, post-generation, we prune near-duplicate paths using a cosine similarity threshold $\gamma=0.85$. This filtering step preserves semantic diversity and mitigates the over-representation of similar reasoning chains.

\textbf{Validate:}
The Validate phase provides a quantitative score for each generated reasoning chain based on its recommendation quality, which serves as the empirical basis for subsequent mining.

Since the LLM outputs textual descriptions rather than Item IDs, we first map the generated content to the item catalog via semantic retrieval. For instance, to map a generated description like ``wireless controller'' to the actual item ``Sony DualSense,'' we encode both the description and titles using Qwen3-8B-Embedding~\cite{zhang2025qwen3}. We retrieve the top-1 item based on cosine similarity to form the unique predicted set $\widehat{Y}_{20}(C)$.

We evaluate the chain using Recall@20. The performance of a chain $C$ is calculated as:
$\operatorname{Recall@20}(C)
=\mathbb{I}[i_{T+1}\in \widehat{Y}_{20}(C)].$
To assess generalized quality, we define Score(C) as the expected Recall@20 across the user distribution. This score measures how consistently a chain produces high-quality recommendations.
\[
\operatorname{Score}(C)
\;=\;
\mathbb{E}_S\!\big[\operatorname{Recall@20}(C)\big].
\]

This score provides a direct empirical estimate of a chain's predictive value, prioritizing reasoning paths that actually lead to accurate recommendation outcomes.

\textbf{Mine:}
The Mine phase abstracts a single, optimal reasoning pattern from the candidate reasoning chains.

The process begins by extracting the \textit{logic flow description} (i.e., the text in the \texttt{<cot\_path>} tag) from each generated response. These descriptions serve as candidate representations of the Reasoning Pattern $P$ defined in Section~\ref{subsec:preliminaries}. We transform these path descriptions into a dense embedding space using a pre-trained sentence encoder (e.g., Qwen3-Embedding~\cite{zhang2025qwen3}). Within this space, we perform K-Means clustering to group semantically similar logic flows, where the optimal number of clusters is determined via the Silhouette Coefficient.

Each cluster represents a candidate pattern $P$, and let $\mathcal{C}_P$ denote its associated set of reasoning chains (i.e., the detailed instantiations in the \texttt{<reason>} tag). We define the \textit{Quality} of a pattern as the average effectiveness of its chains:
\[\text{Quality}(P) = \frac{1}{|\mathcal{C}_P|} \sum\limits_{C \in \mathcal{C}_P} \mathrm{Score}(C).\]
We select the optimal pattern $P^*$ maximizing Quality.
 
Upon selecting $P^*$, we extract it as a symbolic, generalizable template. This transformation both provides a human-interpretable artifact for qualitative analysis and serves as a robust instruction set for the subsequent Structured Distillation phase. We achieve this abstraction through a two-stage, LLM-driven synthesis process.

First, we select the top-$N$ (e.g., $N=10$) chains with the highest cosine similarity to the pattern's semantic centroid. These exemplars serve as a reliable basis for abstraction. Second, we compile these exemplars into a meta-prompt that directs a powerful LLM to synthesize the shared logical structure. The process culminates in an Optimal CoT Template that captures the core reasoning logic of the discovered pattern. An instantiation of the Mine prompt and the resulting synthesized template is shown in Fig.~\ref{fig:mining_prompt}.

\subsection{How to Transfer: Structure-Preserving Integration}
To transfer the discovered pattern $P^*$ without structural loss, we employ a two-stage process. First, offline Structured Distillation (Section~\ref{subsubsec:str_dis}) materializes the pattern into step-wise reasoning matrices  $\{\mathbf{H}_i\}$. Second, online Order-Preserving Fusion (Section~\ref{subsubsec:order_fusion}) integrates the retrieved $\mathbf{H}_i$ with the backbone model while preserving sequential dependencies.

\subsubsection{Structured Distillation}\label{subsubsec:str_dis}
This stage aims to preserve the step-wise logical flow of the discovered pattern without incurring the high latency of online LLM inference. We achieve this through a structured teacher-student distillation framework.

We leverage the optimal template to guide a powerful teacher LLM (e.g., DeepSeek-R1) in generating structured reasoning chains. For each user sequence $S$ in our training corpus, the teacher model produces template-consistent reasoning $C=(c_1, c_2, \ldots, c_K)$, creating training pairs $\{(S_i, C_i)\}_{i=1}^N$ where the student learns to generate structured reasoning given user sequences as input. An example synthesized instance is provided in Fig.~\ref{fig:full_synthesis_example}.

A smaller, more efficient student model (Qwen3-8B~\cite{yang2025qwen3}) is fine-tuned on this synthetic dataset. This step is critical for scalability, as it allows us to generate reasoning chains for the entire user corpus at a fraction of the computational cost of the teacher model.

We apply the distilled student model to generate reasoning chains for all data splits. For each sequence $S_i$, we feed it through the fine-tuned student model to produce a corresponding reasoning chain $C_i=(c_{i,1}, c_{i,2}, \ldots, c_{i,K})$.

For each generated reasoning step $c_{i,j}$, we extract a dense embedding using a pre-trained sentence encoder (e.g., Qwen3-8B-Embedding~\cite{zhang2025qwen3}). This process transforms textual reasoning steps into fixed-dimensional embeddings, where $\mathbf{e}_{i,j} \in \mathbb{R}^D$ represents the embedding for the $j$-th reasoning step of sequence $S_i$:
\[
\mathbf{e}_{i,j} = \text{SentenceEncoder}(c_{i,j}), \quad j = 1, 2, \ldots, K.
\]
The step-wise embeddings for each sequence are then assembled into a structured representation matrix $\mathbf{H}_i \in \mathbb{R}^{K \times D}$, which preserves the sequential structure of the reasoning steps:
\[
\mathbf{H}_i = [\mathbf{e}_{i,1}; \mathbf{e}_{i,2}; \ldots; \mathbf{e}_{i,K}].
\]

All structured embedding matrices $\{\mathbf{H}_i\}_{i=1}^N$ are computed and stored offline and retrieved at inference time. This decoupling strategy eliminates the latency bottleneck of real-time generation, as the complex reasoning is reduced to a fast retrieval operation.

\subsubsection{Order-Preserving Fusion}\label{subsubsec:order_fusion}
This stage integrates pre-computed step-wise matrices $\mathbf{H}_i$ with backbone recommendation models using a lightweight, model-agnostic fusion architecture. This online component prioritizes serving efficiency while preserving the sequential structure critical for reasoning effectiveness.

For each user sequence, we retrieve its corresponding reasoning matrix, $\mathbf{H}_i \in \mathbb{R}^{K \times D}$, from the offline repository during inference. An adapter module then projects these reasoning embeddings into the target model's representation space:
\begin{equation}
\mathbf{z}_{i,j} = \text{LayerNorm}(\mathbf{W}_{\text{proj}} \mathbf{e}_{i,j} + \mathbf{b}_{\text{proj}})
\end{equation}
where $\mathbf{e}_{i,j} \in \mathbb{R}^D$ is the $j$-th step embedding from $\mathbf{H}_i$, $\mathbf{W}_{\text{proj}} \in \mathbb{R}^{d_{\text{item}} \times D}$ projects to the backbone's item embedding dimension, and $\mathbf{z}_{i,j} \in \mathbb{R}^{d_{\text{item}}}$ is the adapted representation.

To strictly enforce the logical order defined by the pattern (e.g., ensuring "Analysis" precedes "Recommendation"), we augment each projected embedding with learnable positional encodings:
\begin{equation}
\mathbf{z}_{i,j}^{\text{pos}} = \mathbf{z}_{i,j} + \mathbf{P}_j
\end{equation}
where $\mathbf{P}_j \in \mathbb{R}^{d_{\text{item}}}$ are position embeddings that encode each step's role within the reasoning sequence.

We employ cross-attention to allow each sequence position to selectively attend to relevant reasoning steps. Let $\mathbf{h}_{\text{seq}}^{(i)} = [\mathbf{h}_1^{(i)}, \mathbf{h}_2^{(i)},$ \ldots,  $\mathbf{h}_T^{(i)}]\in \mathbb{R}^{T \times d_{\text{item}}}$ denote the backbone model's embeddings for the user sequence, and $\mathbf{Z}_i^{\text{pos}} = [\mathbf{z}_{i,1}^{\text{pos}}; \ldots ; \mathbf{z}_{i,K}^{\text{pos}}] \in \mathbb{R}^{K \times d_{\text{item}}}$ represent the projected CoT embeddings with positional encoding. In this cross-attention, the sequence embeddings serve as queries, while the reasoning steps act as both keys and values. The cross-attention mechanism computes attended reasoning representations for each sequence position:
\begin{equation}
\mathbf{A}_i = \text{Attention}(\mathbf{h}_{\text{seq}}, \mathbf{Z}_i^{\text{pos}}, \mathbf{Z}_i^{\text{pos}})
\end{equation}

The attention output is then integrated with the original sequence using adaptive gating:
\begin{align}
\mathbf{g}_i &= \sigma(\mathbf{W}_g[\mathbf{h}_{\text{seq}}; \mathbf{A}_i] + \mathbf{b}_g) \\
\mathbf{E}_{\text{fused}}^{(i)} &= \text{LayerNorm}(\mathbf{g}_i \odot \mathbf{h}_{\text{seq}}^{(i)} + (1 - \mathbf{g}_i) \odot \mathbf{A}_i)
\end{align}
Here, $[\mathbf{h}_{\text{seq}}^{(i)}; \mathbf{A}_i]$ represents concatenation along the feature dimension, and the final layer normalization is applied to the gated output.

To align the reasoning space with the recommendation objective, we employ a contrastive learning objective. Crucially, we calculate the InfoNCE loss~\cite{oord2018representation} specifically between the final reasoning step $\mathbf{z}_{i,K}^{pos}$ (which corresponds to "Item Recommendation" in our pattern) and the target item embedding $\mathbf{v}_{\text{target}}$:
\begin{equation}
\mathcal{L}_{\text{InfoNCE}} = -\log \frac{\exp(\text{sim}(\mathbf{z}_{i, K}^{pos}, \mathbf{v}_{\text{target}}) / \tau)}{\sum_{j=1}^{B} \exp(\text{sim}(\mathbf{z}_{i,K}^{pos}, \mathbf{v}_j) / \tau)}
\end{equation}
The term $\text{sim}(\cdot, \cdot)$ represents cosine similarity, $\tau$ is the temperature parameter, and $B$ is the batch size, with $\{\mathbf{v}_j\}_{j=1}^B$ including the target item and negative samples from other batch items.

Simultaneously, we optimize the primary recommendation task. Let $p(i_{T+1} | \mathbf{E}_{\text{fused})}^{(i)}$ denote the predicted probability of the ground-truth next item $i_{T+1}$ given the fused sequence representation. The general recommendation loss is defined as:
\begin{equation} 
\mathcal{L}_{\text{rec}} = -\log p(i_{T+1} | \mathbf{E}_{\text{fused}}^{(i)})
\end{equation}

The full training objective combines the recommendation loss with the contrastive alignment loss:
\begin{equation}
\mathcal{L}_{\text{total}} = \mathcal{L}_{\text{rec}} + \lambda \mathcal{L}_{\text{InfoNCE}}
\end{equation}
where $\lambda$ controls the contrastive term.

This structured integration architecture preserves the step-wise nature of CoT reasoning, allowing downstream models to leverage both the progressive reasoning flow and the final recommendation-oriented representations to improve prediction accuracy.

\begin{table*}[t]
  \caption{Overall performance comparison between baselines and our proposed SCoTER. Specifically, SCoTER is applied on top of TIGER and HSTU (denoted as SCoTER-T and SCoTER-H), where ``RI'' indicates the relative improvement compared to the respective backbones. The best and second-best results are highlighted in \textbf{bold} and \underline{underlined}, respectively.}
  \label{tab:overall_performance}
  \centering
  \small 
  \setlength{\tabcolsep}{4pt} 
  \renewcommand{\arraystretch}{1.2} 
  \begin{tabular}{@{}llccccc|ccc|ccc@{}}
    \toprule
    & & \multicolumn{5}{c|}{Traditional Baselines} & \multicolumn{3}{c|}{TIGER Backbone} & \multicolumn{3}{c}{HSTU Backbone} \\
    \cmidrule(lr){3-7}\cmidrule(lr){8-10}\cmidrule(l){11-13}
    Dataset & Metric & LightGCN & Caser & HGN & Bert4Rec & SASRec & TIGER & \textbf{SCoTER-T} & \textbf{RI} & HSTU & \textbf{SCoTER-H} & \textbf{RI} \\
    \midrule
    \multirow{4}{*}{\textbf{Beauty}}
      & R@5   & 0.0228 & 0.0279 & 0.0344 & 0.0203 & 0.0387 & 0.0392 & 0.0434 & 10.71\% & \underline{0.0567} & \textbf{0.0737} & 29.98\% \\
      & R@10  & 0.0421 & 0.0456 & 0.0564 & 0.0347 & 0.0605 & 0.0594 & 0.0656 & 10.44\% & \underline{0.0930} & \textbf{0.1187} & 27.63\% \\
      & N@5   & 0.0136 & 0.0172 & 0.0214 & 0.0124 & 0.0249 & 0.0257 & 0.0276 & 7.39\% & \underline{0.0376} & \textbf{0.0472} & 25.53\% \\
      & N@10  & 0.0198 & 0.0229 & 0.0284 & 0.0137 & 0.0318 & 0.0321 & 0.0347 & 8.10\% & \underline{0.0502} & \textbf{0.0621} & 23.71\% \\
    \midrule
    \multirow{4}{*}{\textbf{Instrument}}
      & R@5   & 0.0757 & 0.0770 & 0.0813 & 0.0671 & 0.0857 & 0.0865 & 0.0908 & 4.97\% & \underline{0.1009} & \textbf{0.1111} & 10.11\% \\
      & R@10  & 0.1010 & 0.0995 & 0.1048 & 0.0822 & 0.1083 & 0.1062 & 0.1110 & 4.52\% & \underline{0.1394} & \textbf{0.1573} & 12.84\% \\
      & N@5   & 0.0472 & 0.0639 & 0.0668 & 0.0560 & 0.0715 & \underline{0.0736} & \textbf{0.0765} & 3.94\% & 0.0635 & 0.0669 & 5.35\% \\
      & N@10  & 0.0554 & 0.0711 & 0.0774 & 0.0608 & 0.0788 & 0.0799 & \textbf{0.0829} & 3.75\% & 0.0769 & \underline{0.0823} & 7.02\% \\
    \midrule
    \multirow{4}{*}{\textbf{Sports}}
      & R@5   & 0.0098 & 0.0116 & 0.0189 & 0.0115 & 0.0233 & 0.0233 & 0.0260 & 11.59\% & \underline{0.0418} & \textbf{0.0584} & 39.71\% \\
      & R@10  & 0.0184 & 0.0194 & 0.0313 & 0.0191 & 0.0350 & 0.0379 & 0.0406 & 7.12\% & \underline{0.0674} & \textbf{0.0889} & 31.90\% \\
      & N@5   & 0.0061 & 0.0072 & 0.0120 & 0.0075 & 0.0154 & 0.0150 & 0.0161 & 7.33\% & \underline{0.0274} & \textbf{0.0387} & 41.24\% \\
      & N@10  & 0.0087 & 0.0097 & 0.0159 & 0.0099 & 0.0192 & 0.0197 & 0.0209 & 6.09\% & \underline{0.0361} & \textbf{0.0487} & 34.90\% \\
    \midrule
    \multirow{4}{*}{\textbf{Yelp}}
      & R@5   & 0.0248 & 0.0150 & 0.0186 & 0.0186 & 0.0183 & 0.0241 & 0.0258 & 7.05\% & \underline{0.0257} & \textbf{0.0356} & 38.52\% \\
      & R@10  & 0.0403 & 0.0263 & 0.0326 & 0.0291 & 0.0296 & 0.0385 & 0.0406 & 5.45\% & \underline{0.0438} & \textbf{0.0567} & 29.45\% \\
      & N@5   & 0.0156 & 0.0099 & 0.0115 & 0.0115 & 0.0116 & 0.0158 & 0.0174 & 10.13\% & \underline{0.0177} & \textbf{0.0242} & 36.72\% \\
      & N@10  & 0.0207 & 0.0134 & 0.0159 & 0.0159 & 0.0152 & 0.0204 & 0.0222 & 8.82\% & \underline{0.0245} & \textbf{0.0315} & 28.57\% \\
    \bottomrule
  \end{tabular}
\end{table*}

\section{Experiments}
In this section, we conduct a series of experiments to answer the following research questions:

\textbf{RQ1 (Overall Performance):} How does SCoTER perform compared to existing recommendation models?

\textbf{RQ2 (Pattern Discovery):} How effective is our automated GVM pipeline compared to manual CoT templates? Are the discovered patterns 
dataset-specific?

\textbf{RQ3 (Structure Preservation):} How do the structure-preserving components contribute to performance?

\textbf{RQ4 (Robustness \& Sensitivity):} Is SCoTER robust across different user/item groups and hyperparameter settings?

\subsection{Experimental Setup}

\textbf{Datasets.} We conduct experiments on four widely used datasets: three subsets of the Amazon Product Reviews dataset \cite{he2016ups, ni2019justifying} (Beauty, Instruments, and Sports) and the Yelp dataset. Table~\ref{tab:dataset} summarizes the statistics of these datasets. Following previous work \cite{zhou2020s3}, we process the data to enforce a 5-core density, removing all users and items with fewer than five interactions. All user sequences are then normalized to a uniform length of 20 through padding or truncation, preserving their most recent interactions. For evaluation, we use the leave-one-out protocol: each user's final interaction is designated for testing, the penultimate one for validation, and the remaining interactions are used for training

\textbf{Baselines.} We compare our proposed method, Structured Chain-of-Thought Recommendation (SCoTER), against a comprehensive suite of representative baselines that span different paradigms:

\begin{itemize}
    \item \textbf{LightGCN}~\cite{he2020lightgcn}: A graph convolutional network that captures collaborative signals via neighborhood aggregation.
    \item \textbf{Caser}~\cite{tang2018personalized}: A sequential model that employs convolutional neural networks to capture local sequential patterns.
    \item \textbf{HGN}~\cite{ma2019hierarchical}: A sequential model that utilizes a hierarchical gating network to adaptively integrate a user's long- and short-term preferences.
    \item \textbf{SASRec}~\cite{kang2018self}: A sequential model that uses a self-attention mechanism to capture long-range dependencies and dynamic user preferences.
    \item \textbf{Bert4Rec}~\cite{sun2019bert4rec}: A sequential model that uses a deep bidirectional self-attention mechanism to model user sequences.
    \item \textbf{TIGER}~\cite{rajput2023recommender}: A generative model that represents items as discrete token sequences, enabling recommendation through autoregressive decoding. We select TIGER as the backbone for our method due to its strong generative performance and architectural compatibility with reasoning integration.
    \item \textbf{HSTU}~\cite{zhai2024actions}: A sequential model utilizing hierarchical sequential transduction units. We augment the original HSTU with semantic ID inputs and a dedicated SID decoder to enable generative retrieval (details in Appendix~\ref{app:backbones}). This enhanced version serves as a stronger baseline and backbone for SCoTER-H.
    \item \textbf{SCoTER:} Our proposed framework. To demonstrate its universality, we implement two variants with different backbones: SCoTER-T (built upon TIGER) and SCoTER-H (built upon HSTU). Both variants enhance their respective backbones by integrating structured Chain-of-Thought reasoning, as detailed in Section~\ref{sec:method}.
\end{itemize}

\textbf{Evaluation Protocol.} Performance is evaluated using two standard top-K ranking metrics: Recall@K and NDCG@K.  Following common practice, we report the main results for $K \in \{5, 10\}$. To ensure a fair evaluation and avoid sampling bias, we perform a full ranking over the entire item catalog for each user.

\textbf{Implementation Details.} For traditional methods, we follow standard implementations with hyperparameters tuned on validation sets. For generative backbones, the tokenizer employs RQ-VAE~\cite{lee2022autoregressive, zeghidour2021soundstream} for discrete semantic encoding with 4 codebooks, each containing 256 embeddings of dimension 32. Semantic inputs to RQ-VAE are derived from the embeddings of item titles and descriptions processed by Qwen3-8B-Embedding~\cite{zhang2025qwen3}. During inference, we use a beam size of 20 to balance recommendation quality and efficiency. For SCoTER, we integrate the pre-computed offline reasoning embeddings $H_i$ via multi-head cross-attention. We adopt distinct configurations to optimize each variant: SCoTER-T utilizes 6 cross-attention heads and is trained for 200 epochs. In contrast, SCoTER-H employs 4 cross-attention heads and is trained for 300 epochs to ensure convergence. Training uses the Adam optimizer with learning rate $1 \times 10^{-4}$, weight decay $5 \times 10^{-5}$, and a batch size of 512. The contrastive learning weight $\lambda$ is set to 0.1 for both variants. Additional hyperparameter details are provided in the Appendix~\ref{app:train_details}.

\subsection{Overall Performance (RQ1)}
To assess the effectiveness of our reasoning transfer framework, we compare SCoTER against strong baselines. The comprehensive results are presented in Table~\ref{tab:overall_performance}.

\textbf{Consistent Improvements.} SCoTER delivers consistent gains across all datasets and metrics. SCoTER-T achieves 3.75\%--11.59\% relative improvement over TIGER, while SCoTER-H yields larger gains of 5.35\%--41.24\% over the stronger HSTU backbone, achieving state-of-the-art results on most metrics.

\textbf{Dataset-wise Analysis.} The most substantial improvements appear on Beauty and Sports datasets, where SCoTER-H boosts Recall@5 by 29.98\% and 39.71\%, respectively. On Instrument, HSTU's NDCG metrics are lower than TIGER's, possibly due to the smaller item vocabulary favoring TIGER's semantic tokenization. Nonetheless, SCoTER consistently improves both backbones despite these intrinsic differences. The gains are generally more pronounced for top-5 than top-10 metrics, suggesting that structured reasoning particularly benefits precision-critical scenarios.

\textbf{Backbone-Agnostic Design.} The consistent improvements observed across both TIGER  and HSTU  validate that SCoTER is backbone-agnostic, injecting effective reasoning capabilities orthogonal to the underlying base architecture.

\subsection{Automated Pattern Discovery (RQ2)}

To evaluate our automated GVM pipeline, we compare the discovered pattern against manual CoT templates (Two-step, Three-step, Five-step) representing increasing reasoning depth. Specifically, the Two-step template follows Rec-GPT~\cite{yi2025recgpt}. The discovered templates for each dataset are provided in Appendix~\ref{sec:appendix_prompts}.

\textbf{GVM vs. Manual Templates.} As shown in Figure~\ref{fig:perform_of_different_cot}, our discovered pattern consistently outperforms all manual templates. On TIGER, SCoTER achieves 10.44\% gain in Recall@10, surpassing Five-step (7.41\%) by a clear margin. On HSTU, SCoTER delivers 27.63\% gain versus 21.72\% for Five-step. The results confirm that our data-driven approach identifies non-obvious logic pathways aligned with the data distribution, breaking through the ceiling of manual template design.

\textbf{Clustering vs. Top-K Selection.} We also compare our clustering-based pattern mining against a Top-K selection baseline ($K=10$) that directly selects high-scoring individual chains without clustering. As shown in Table~\ref{tab:pattern_analysis}, the Top-K selection leads to noticeable performance degradation. This stems from the fact that individual chain selection lacks generalizability---patterns derived from single instances fail to capture the shared reasoning structure across diverse user sequences.

\textbf{Cross-Dataset Pattern Transfer.} To investigate whether discovered patterns are dataset-specific, we transfer patterns across datasets. As shown in Table~\ref{tab:cross_dataset}, cross-dataset transfer generally degrades performance. Notably, while the general Beauty pattern transfers reasonably well to Sports, the Sports pattern performs poorly on Beauty because it is highly domain-specific (e.g., emphasizing athletic attributes). This asymmetry confirms that effective reasoning patterns are domain-dependent, validating per-dataset discovery via GVM.

\begin{table}[t]
\caption{Pattern discovery analysis on HSTU backbone.}
  \label{tab:pattern_analysis}
  \centering
  \small
  \begin{tabular}{llc}
    \toprule
    Variant (on Beauty) & R@10 &N@10  \\
    \midrule
    Clustering (SCoTER) & 0.1187 & 0.0621\\
    Top-K Selection & 0.1124 & 0.0593 \\
    \bottomrule
  \end{tabular}
\end{table}

\begin{table}[t]
  \caption{Cross-dataset pattern transfer on HSTU backbone.}
  \label{tab:cross_dataset}
  \centering
  \small
  \begin{tabular}{llcc}
    \toprule
    Pattern Source & Target & R@10 &  N@10 \\
    \midrule
    Beauty & Beauty & 0.1187 (-)  & 0.0621 (-)\\
    Sports & Beauty & 0.1057 ($\downarrow 10.95\%$) & 0.0555 ($\downarrow 10.63\%$) \\
    \midrule
    Sports & Sports & 0.0889 (-) & 0.0487 (-)  \\
    Beauty & Sports & 0.0860 ($\downarrow 3.26\%$) & 0.0469 ($\downarrow 3.70\%$) \\
    \bottomrule
  \end{tabular}
\end{table}

\begin{figure}[t]
\centering
\includegraphics[width=1.0\linewidth]{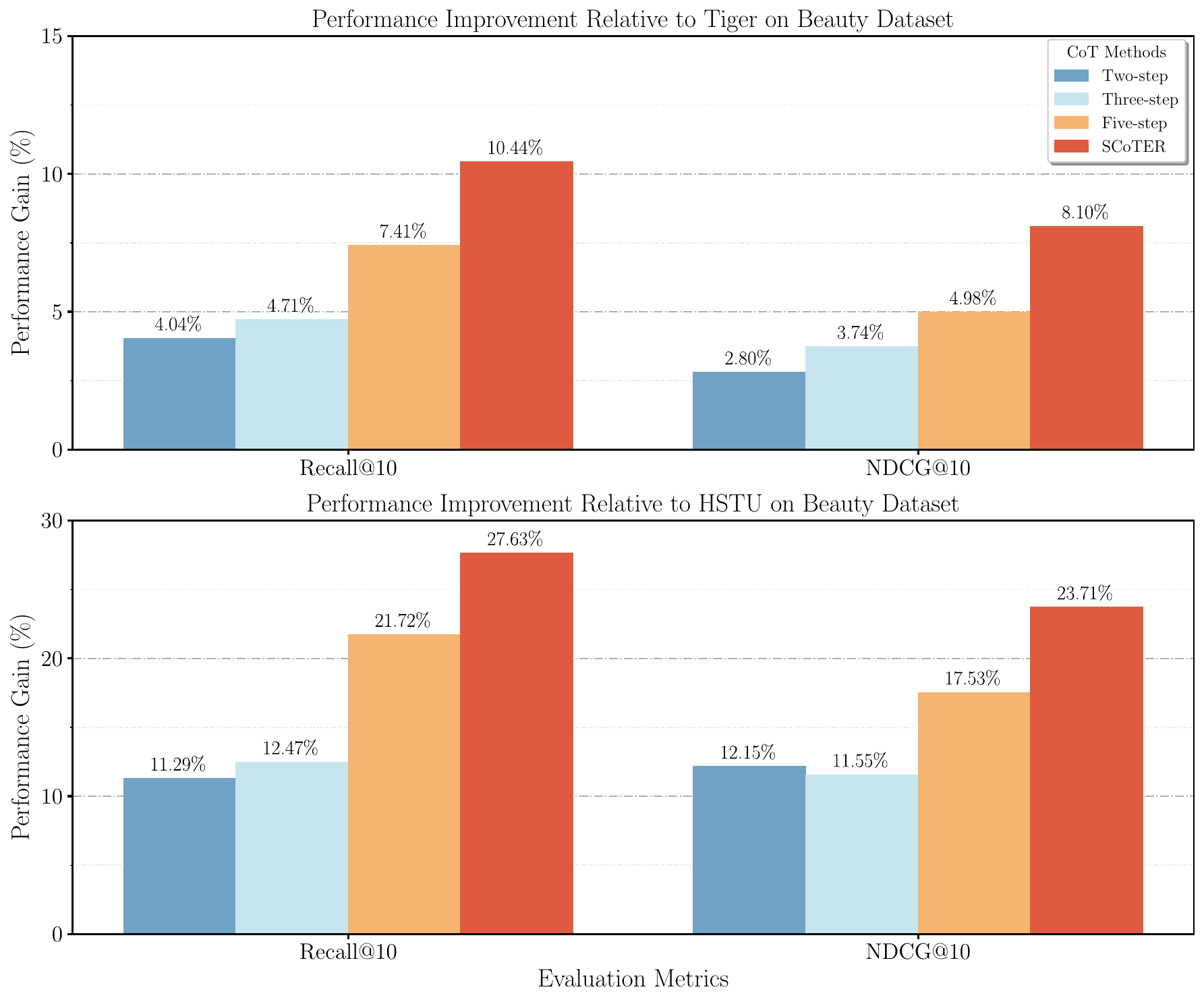}

\caption{Performance improvement relative to TIGER (top) and HSTU (bottom) backbones on the Beauty dataset. We compare manual CoT templates (Two-step, Three-step, Five-step) against our automatically discovered pattern (SCoTER).}
\label{fig:perform_of_different_cot}
\end{figure}

\begin{table}[t]
  \caption{Ablation study of SCoTER-H on Beauty and Sports datasets. We report R@10 and N@10.}
  \label{tab:ablation_study_on_hstu}
  \centering
  \small
  \setlength{\tabcolsep}{3.5pt} 
  \renewcommand{\arraystretch}{1.2}
  \begin{tabular}{lcccc}
    \toprule
    \multirow{2}{*}{Variant} & \multicolumn{2}{c}{Beauty} & \multicolumn{2}{c}{Sports} \\
    \cmidrule(lr){2-3} \cmidrule(lr){4-5}
    & R@10 & N@10 & R@10 & N@10 \\
    \midrule
    \textbf{SCoTER-H (Full)} & \textbf{0.1187} & \textbf{0.0621} & \textbf{0.0889} & \textbf{0.0487} \\
    \midrule
    w/o Position & 0.1143 & 0.0605 & 0.0882 & 0.0479 \\
    w/o Contrastive & 0.1091 & 0.0590 & 0.0842 & 0.0454 \\
    w/o Step-wise Embedding & 0.1139 & 0.0607 & 0.0864 & 0.0468 \\
    w/o Step-wise (Mean) & 0.1105 & 0.0597 & 0.0856 & 0.0462 \\
    w/o Step-wise (Last) & 0.1030 & 0.0554 & 0.0827 & 0.0454 \\
    w/o Student Distill &  0.1008 & 0.0537 & 0.0757 & 0.0407\\
    \midrule
    HSTU (Backbone) & 0.0930 & 0.0502 & 0.0674 & 0.0361 \\
    \bottomrule
  \end{tabular}
  \vspace{-0.5pt}
\end{table}

\subsection{Effectiveness of Structure Preservation (RQ3)}
To validate each component's contribution, we conduct ablation studies on Beauty and Sports using SCoTER-H. Results are shown in Table~\ref{tab:ablation_study_on_hstu}. We also find that ablations on TIGER show consistent trends (see Appendix~\ref{append:inter_ab_on_tiger}).

\textbf{Step-wise Fusion is Critical.} Replacing step-wise attention 
with compressed representations causes significant degradation. Using only the last token (w/o Step-wise Last) leads to 13.23\% drop in Recall@10 on Beauty, while mean pooling (w/o Step-wise Mean) results in 6.91\% drop, and encoding the entire CoT as a single vector (w/o Step-wise Embedding) causes 4.04\% drop. This confirms that intermediate reasoning steps contain valuable signals; discarding the logical progression reduces fine-grained matching ability.

\textbf{Alignment and Position Matter.} Removing contrastive alignment (w/o Contrastive) causes 8.09\% drop in Recall@10 on Beauty, confirming that reasoning must be explicitly aligned with user preferences. Removing positional encoding (w/o Position) degrades Beauty by 3.71\%, indicating that step order is important for logical deduction. Interestingly, the position impact on Sports is marginal, suggesting that for certain domains, semantic content dominates over strict order.

\textbf{Student Distillation is Essential.} Using the original Qwen3-8B~\cite{yang2025qwen3} without task-specific fine-tuning (w/o Student Distill) results in significantly lower performance. Specifically, on the Beauty dataset, this variant yields only an 8.39\% improvement in Recall@10 over the HSTU backbone, compared to the 27.63\% gain achieved with distillation. This demonstrates that task-specific adaptation is more critical than raw model capability.

\subsection{Robustness and Sensitivity Analysis (RQ4)}

To further analyze the robustness of our framework, we split users and items into different subgroups to gain deeper insights into SCoTER's behavior under varying data sparsity conditions.

\textbf{Setup.} For users, following~\cite{tang2025think}, we divide them into four equal-sized groups based on sequence length. Group~0 represents users with the sparsest interactions, while higher group numbers indicate longer sequences. For items, consistent with previous studies~\cite{yang2023debiased, tang2024towards}, we categorize them into four quartiles based on interaction frequency. Group~0 consists of the least popular items, whereas higher group numbers indicate more popular items. 

\textbf{Results.} Figure~\ref{fig:group_analysis_beauty} illustrates the relative improvement of SCoTER-H over HSTU on the Beauty dataset. We observe consistent performance gains across all user and item groups, demonstrating that our reasoning framework is effective regardless of data sparsity.

\begin{figure}[t]
\centering
\includegraphics[width=1\linewidth]{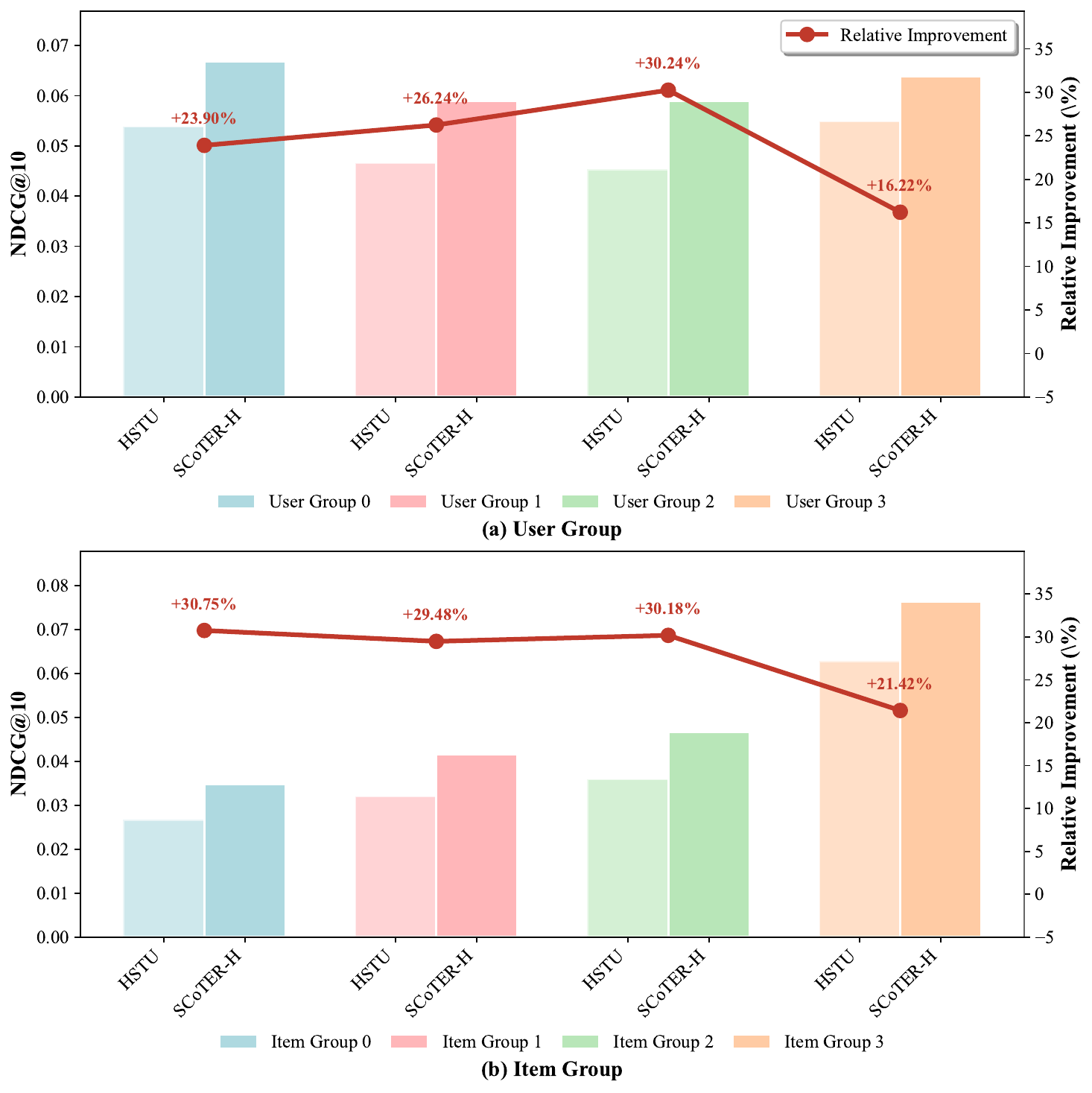}
\caption{Performance comparison across different user and item groups on the Beauty dataset. Groups are divided based on interaction frequency, where Group 0 represents the sparsest scenarios and Group 3 represents the densest.}
\label{fig:group_analysis_beauty}
\vspace{-2mm}
\end{figure}

\textbf{Sparse Scenarios Benefit Most.} Notably, SCoTER-H achieves substantial improvements in sparse scenarios. In User Group~0, which suffers from limited historical data, the model yields a 23.90\% improvement. This suggests that when collaborative signals are weak, the explicit reasoning chain helps the model deduce user intent effectively. Similarly, for Item Group~0 (tail items), we observe a remarkable 30.75\% improvement. This indicates that incorporating semantic reasoning helps alleviate the popularity bias often found in generative recommenders, allowing the model to better surface relevant but less popular items.

\textbf{Sweet Spot in Middle Groups.} We also notice that the relative improvement peaks in middle groups (e.g., User Group~2 shows +30.24\%), suggesting a sweet spot where the model effectively combines sufficient interaction history with reasoning logic to maximize performance. Even in the densest groups (Group~3), where the backbone is already strong, SCoTER maintains a significant uplift (+16.22\% for users, +21.42\% for items), confirming its additive value across all scenarios.

\begin{figure}[t]
\centering
\includegraphics[width=1\linewidth]{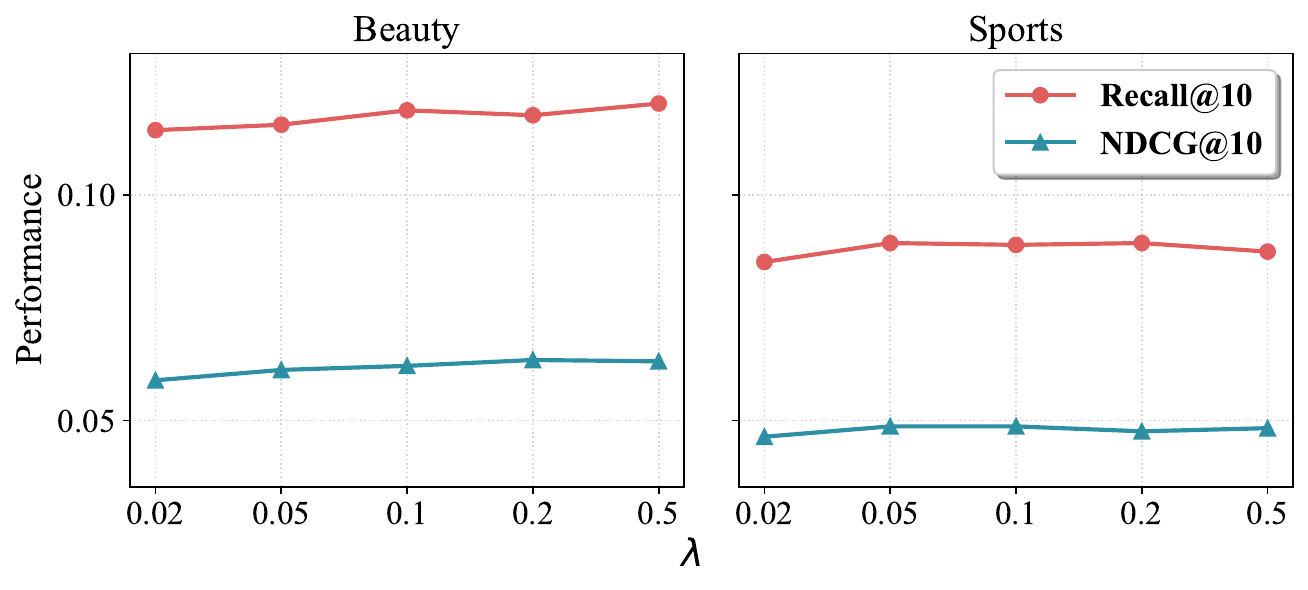}
\caption{Impact of the contrastive learning weight $\lambda$ on Beauty and Sports datasets. The results indicate that SCoTER is relatively robust to variations in $\lambda$.}
\label{fig:lambda_sensitivity}
\end{figure}

\textbf{Sensitivity to Contrastive Weight $\lambda$.}
Figure~\ref{fig:lambda_sensitivity} shows the impact of $\lambda$ on Beauty and Sports. Performance improves as $\lambda$ increases from 0.02 to 0.1, confirming that contrastive alignment is beneficial. The performance remains stable within 0.05--0.2 and peaks around 0.1. Further increasing $\lambda$ to 0.5 causes slight degradation, indicating that excessive alignment weight may overshadow the primary recommendation objective.

Efficiency analysis is provided in Appendix~\ref{app:efficiency}.

\begin{table}[ht]
  \centering 
  \caption{Relative improvement of our online A/B testing on the Tencent Advertising Platform.} 
  \label{tab:online_results} 
  
  \begin{tabular}{lr} 
    \toprule 
    \textbf{Online Metric} & \textbf{Relative Lift} \\ 
    \midrule 
    
    GMV (Overall) & +2.14\% \\
    GMV (Sparse Users) & +4.10\% \\
    GMV (Dense Users) & +1.49\% \\
    
    \midrule 
    
    Negative Feedback Rate & -0.24\% \\
    "Not Interested" Rate & -0.25\% \\
    
    \bottomrule 
  \end{tabular}
  \vspace{-0.2pt}
\end{table}

\subsection{Online A/B Test}

We validated SCoTER through deployment on the Tencent Advertising Platform. A one-week online A/B test with 5\% traffic compared SCoTER against the production baseline, using Gross Merchandise Value (GMV) as the primary metric.

As reported in Table \ref{tab:online_results}, SCoTER achieved a +2.14\% lift in overall GMV. Stratified analysis reveals that gains are most pronounced for sparse users (+4.10\%) compared to dense users (+1.49\%), highlighting SCoTER's potential to mitigate data sparsity---consistent with our offline findings. User experience also improved: negative feedback rate decreased by 0.24\% and ``not interested'' rate by 0.25\%, indicating better alignment with user preferences.



\section{Conclusion}

In this paper, we identify and address two challenges in applying CoT reasoning to recommendation: discovering effective reasoning patterns beyond brittle and hand-crafted heuristics, and transferring them to efficient models without collapsing their essential stepwise logic under low-latency demands. To tackle these challenges in a unified manner, we propose SCoTER, a novel framework featuring an automated GVM pipeline for pattern discovery and a structure-preserving architecture. The efficacy of this framework is validated through comprehensive experiments and real-world deployment. Comprehensive experiments demonstrate that SCoTER not only consistently outperforms state-of-the-art baselines but also achieves a 2.14\% lift in production GMV. Together, these results establish SCoTER as a systematic and empirically-grounded methodology for integrating structured LLM reasoning into recommender systems.

\bibliographystyle{ACM-Reference-Format}
\bibliography{sample-base}

\appendix

\newpage
\section{Experimental Setup and Hyperparameter}~\label{append:hyper}

\begin{table}[t]
\centering
\caption{Dataset statistics of the evaluation benchmarks. ``AvgLen'' represents the average length of item sequences.}
\label{tab:dataset}
\begin{tabular}{lcccc}
\toprule
Dataset & \#Users & \#Items & \#Interactions & AvgLen \\
\midrule
Beauty      & 22,363 & 12,101 & 198,502 & 8.88 \\
Instruments & 24,772 & 9,922  & 206,153 & 8.32 \\
Sports      & 35,598 & 18,357 & 296,337 & 8.32 \\
Yelp        & 30,431 & 20,033  &  316,354 &  10.40    \\
\bottomrule
\end{tabular}
\end{table}

\begin{figure}[t]
    \centering
    \includegraphics[width=\linewidth]{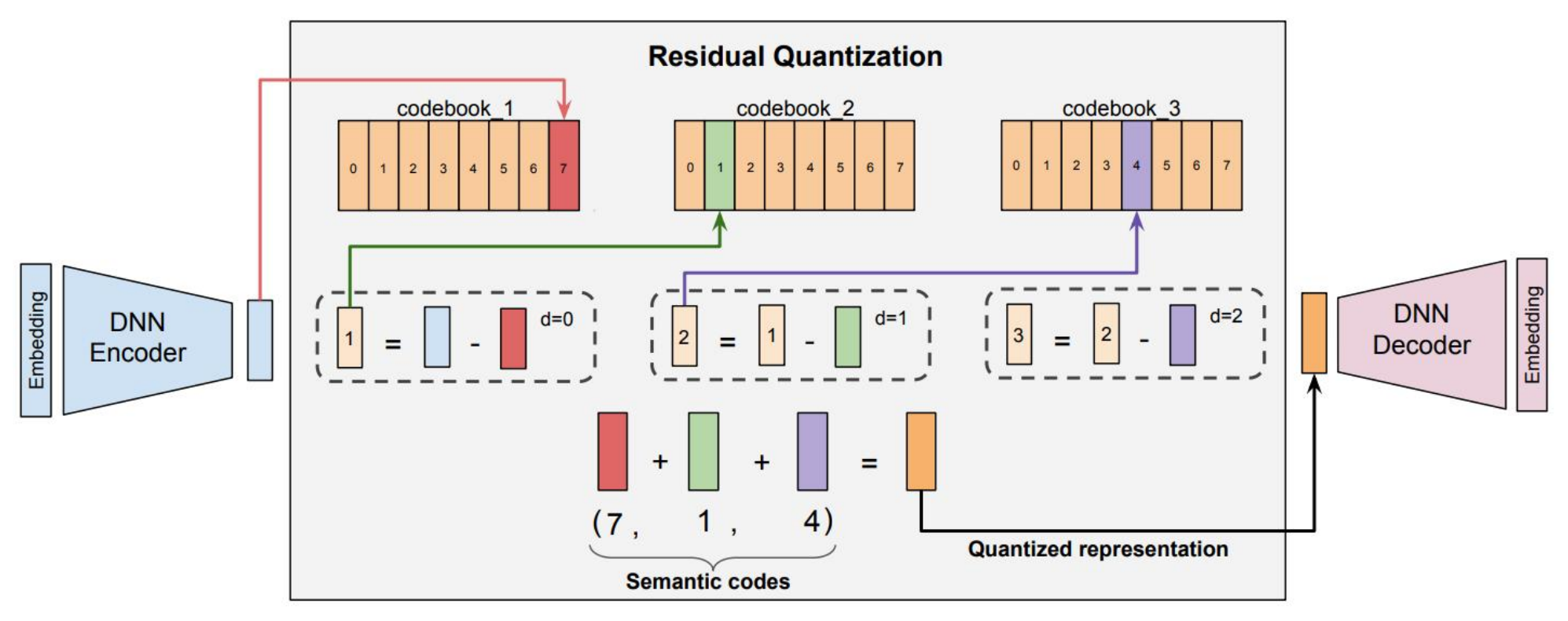}
    \caption{RQ-VAE for Semantic ID quantization via residual vector quantization (adapted from TIGER~\cite{rajput2023recommender}).}
    \label{fig:rqvae}
\end{figure}

\subsection{Backbone Models}
\label{app:backbones}

\paragraph{RQ-VAE for Semantic IDs (RQVAE)}
We use RQ-VAE (Residual-Quantized VAE) to quantize continuous item embeddings into discrete Semantic IDs (SIDs), as illustrated in Fig.~\ref{fig:rqvae} (original figure from TIGER~\cite{rajput2023recommender}).
RQ-VAE applies multi-stage residual quantization: the first codebook quantizes the latent vector, and subsequent codebooks quantize the remaining residuals in a coarse-to-fine manner. The final SID is the ordered tuple of selected code indices (one per codebook). To ensure uniqueness, an extra disambiguation code can be appended for items sharing the same prefix SID.

\begin{table}[t]
\centering
\small
\caption{RQVAE hyperparameters for Semantic ID quantization.}
\label{tab:rqvae_hparams}
\begin{tabular}{ll}
\toprule
\textbf{Component} & \textbf{Setting} \\
\midrule
Latent dim ($e\_dim$) & 32 \\
\# quantizers (levels) & 4 \\
Codebook sizes ($num\_emb\_list$) & [256, 256, 256, 256] \\
Commitment weight ($\beta$) & 0.25 \\
Reconstruction loss & MSE (default), weight=1.0 \\
Optimizer & Adam, lr=$1\text{e-}3$, weight\_decay=$1\text{e-}4$ \\
Batch size / epochs & 1024 / 20000 \\
\bottomrule
\end{tabular}
\end{table}

\paragraph{TIGER (Generative Backbone).}
We adopt TIGER as a generative backbone that represents each item as a sequence of discrete semantic tokens. TIGER adopts a standard T5~\cite{roberts2023scaling} encoder--decoder architecture for generative retrieval.
The encoder has 4 layers with $d_{\text{model}}{=}128$, $d_{\text{ff}}{=}1024$, 6 attention heads ($d_{kv}{=}64$), and dropout 0.1 (ReLU FFN).
The decoder also has 4 layers with the same hidden size/heads, and autoregressively predicts a 4-token semantic-ID sequence with vocabulary size 1025 (1024 codes + padding). Following the standard practice in generative retrieval, we use an RQ-VAE tokenizer with 4 codebooks, each containing 256 embeddings of dimension 32. The semantic inputs to RQ-VAE are derived from item title and description embeddings produced by Qwen3-8B-Embedding. During inference, we use beam search with beam size 20.

\paragraph{HSTU (with Semantic IDs).}
HSTU is a sequential recommendation model based on hierarchical sequential transduction units.
We augment it with a 4-level Semantic ID (SID) system (256 IDs per level, i.e., $4\times256$) and a 3-layer Transformer SID decoder (4 heads, hidden dim 256). For each item, we obtain an item-ID embedding and 4 SID embeddings (one per SID level), and sum them as the final input embedding to the HSTU encoder. The HSTU encoder uses embedding dim 50 with 2 HSTU blocks (1 head; $d_{qk}=50$, $d_v=50$), dropout 0.2, and max sequence length 20.

\subsection{Training Details and Hyperparameters}
\label{app:train_details}

We use Adam with learning rate $1\times 10^{-4}$ for both variants.
We apply early stopping based on validation performance.
The contrastive alignment weight $\lambda$ is set to 0.1 (shared across variants).
Table~\ref{tab:hyperparams_scoter} summarizes the key hyperparameters.

\begin{table}[t]
\centering
\small
\caption{Key hyperparameters for HSTU-based and TIGER-based variants.}
\label{tab:hyperparams_scoter}
\resizebox{\linewidth}{!}{%
\begin{tabular}{lcc}
\toprule
& \textbf{HSTU-backbone} & \textbf{TIGER-backbone} \\
\midrule
\multicolumn{3}{l}{\textit{Backbone}} \\
Encoder embedding dim $d_{\text{item}}$ & 50 & 128 \\
Encoder layers & 2 (HSTU blocks) &  T5 encoder  \\
\midrule
\multicolumn{3}{l}{\textit{Semantic ID / Tokenization}} \\
SID levels / codebooks & 4 & 4 \\
Per-level vocab size & 256 & 256 \\
SID decoder & 3-layer Transformer, 4 heads, dim 256 & T5 decoder \\
\midrule
\multicolumn{3}{l}{\textit{CoT Fusion}} \\
CoT hidden dim $D$ & 4096 & 4096 \\
$d_{\text{item}}$ & 50 & 128 \\
Dropout (fusion) & 0.1 & 0.1 \\
\midrule
\multicolumn{3}{l}{\textit{Optimization / Inference}} \\
Batch size & 512 & 256 \\
Learning rate & $1\times 10^{-4}$ & $1\times 10^{-4}$ \\
Max epochs & 300 & 200 \\
Early stop patience & 20 & 20 \\
Contrastive weight $\lambda$ & 0.1 & 0.1 \\
\bottomrule
\end{tabular}
}
\end{table}

\section{Efficiency Analysis}
\label{app:efficiency}

To evaluate the computational cost of our proposed framework, we compare the training and inference efficiency of SCoTER-H against the backbone model HSTU. The experiments were conducted on the same hardware environment (e.g., a single NVIDIA H20 GPU). Table \ref{tab:efficiency_appendix} details the training time per epoch and the total inference time for the full test set.

\begin{table}[h]
  \caption{Efficiency comparison between HSTU and SCoTER-H. Training time is reported per epoch, and inference time is reported for the full test set. SCoTER eliminates online LLM inference entirely.}
  \label{tab:efficiency_appendix}
  \centering
  \small
  \begin{tabular}{llcccc}
    \toprule
    \multirow{2}{*}{Dataset} & \multirow{2}{*}{Model} & \multicolumn{2}{c}{Training (per epoch)} & \multicolumn{2}{c}{Inference (full test set)} \\
    \cmidrule(lr){3-4} \cmidrule(lr){5-6}
    & & Time (s) & $\Delta$ & Time (s) & $\Delta$ \\
    \midrule
    \multirow{2}{*}{Beauty} 
    & HSTU & 22.98 & -- & 21.81 & -- \\
    & SCoTER-H & 33.22 & +44.56\% & 26.93 & +23.48\% \\
    \midrule
    \multirow{2}{*}{Sports} 
    & HSTU & 37.14 & -- & 34.39 & -- \\
    & SCoTER-H & 87.36 & +135.22\% & 43.20 & +25.62\% \\
    \bottomrule
  \end{tabular}
\end{table}

As shown in Table \ref{tab:efficiency_appendix}, SCoTER-H introduces a noticeable but manageable computational overhead compared to the raw backbone.
\textbf{Training Efficiency.} The training time increases noticeably (e.g., roughly +135.22\% on Sports). This overhead primarily stems from the additional gradient computations required for the reasoning adapter and the contrastive alignment loss. However, considering the significant performance gains, this cost is acceptable.
\textbf{Inference Latency.} For inference, SCoTER-H remains highly efficient. Unlike LLM-based methods that require seconds per generation, SCoTER-H operates in milliseconds, maintaining the high-throughput characteristic of traditional sequential models. The slight increase in inference time is due to the additional step-wise reasoning token generation, but since the reasoning chain is short (fixed steps), the latency remains within a practical range for real-time industrial deployment.

\section{Pseudo-code Details}
\label{app:pseudocode}

\begin{algorithm}[tb]
\caption{Generate-Validate-Mine (GVM) Pipeline for Automated Pattern Discovery}
\label{alg:gvm}
\begin{algorithmic}[1]
\State {\bfseries Input :} User sequences $\mathcal{S}=\{(S_i, i_{T+1}^{(i)})\}_{i=1}^{N}$, LLM $G$,
sentence encoder $\textsc{SentenceEncoder}(\cdot)$, similarity threshold $\gamma$,
candidate cluster set $\mathcal{K}$, item catalog embeddings $\{ \mathbf{v}_m \}_{m\in \mathcal{I}}$,
top exemplar number $N$
\State {\bfseries Output :} Optimal reasoning pattern $P^*$ and Optimal CoT template $T^*$

\State \textcolor{blue}{/* \hspace{0.2em} Phase 1: Generate \hfill */} 
\State $\mathcal{C}_{all} \leftarrow \emptyset$
\For{$i=1$ {\bfseries to} $N$}
    \State $C_i \leftarrow G.\textsc{Generate}(S_i, \texttt{prompt};\ \tau_{gen}=0.7,\ p=0.9)$
    \State Parse $C_i \Rightarrow (p_i \in \texttt{<cot\_path>},\ \texttt{<reason>},\ R_i \in \texttt{<recommendations>})$
    \State $\mathcal{C}_{all} \leftarrow \mathcal{C}_{all} \cup \{(p_i, R_i, i_{T+1}^{(i)})\}$
\EndFor
\State $\mathcal{C}_{div} \leftarrow \textsc{PruneByCosSim}(\mathcal{C}_{all}, \gamma)$
\hfill \textcolor{blue}{$\triangleright$ cosine threshold $\gamma=0.85$}

\State \textcolor{blue}{/* \hspace{0.2em} Phase 2: Validate \hfill */} 
\For{{\bfseries each} $(p, R, i_{T+1}) \in \mathcal{C}_{div}$}
    \State $\widehat{Y}_{20}(C) \leftarrow \emptyset$
    \For{{\bfseries each} textual recommendation $r \in R$} 
        \State $\mathbf{u} \leftarrow \textsc{SentenceEncoder}(r)$ \hfill \textcolor{blue}{$\triangleright$ semantic retrieval}
        \State $\hat{y} \leftarrow \arg\max_{m\in\mathcal{I}} \cos(\mathbf{u}, \mathbf{v}_m)$ \hfill \textcolor{blue}{$\triangleright$ top-1 mapping}
        \State $\widehat{Y}_{20}(C) \leftarrow \widehat{Y}_{20}(C) \cup \{\hat{y}\}$ 
    \EndFor
    \State $\textsc{Score}(C) \leftarrow \mathbb{I}[ i_{T+1} \in \widehat{Y}_{20}(C) ]$
    \hfill \textcolor{blue}{$\triangleright$ Recall@20 scoring}
\EndFor
\State {\footnotesize (Empirical)} $\ \textsc{Score}(C)$ estimates the expected Recall@20 over users.

\State \textcolor{blue}{/* \hspace{0.2em} Phase 3: Mine \hfill */} 
\State $\mathcal{E} \leftarrow \{\ \mathbf{e}(p)=\textsc{SentenceEncoder}(p)\ :\ (p,\cdot,\cdot)\in\mathcal{C}_{div}\ \}$
\hfill \textcolor{blue}{$\triangleright$ embed logic-flow in \texttt{<cot\_path>}}
\State $K_c \leftarrow \arg\max_{k\in\mathcal{K}} \textsc{Silhouette}(\textsc{KMeans}(\mathcal{E}, k))$
\hfill \textcolor{blue}{$\triangleright$ select \#clusters}
\State $\{\mathcal{P}_j\}_{j=1}^{K_c} \leftarrow \textsc{KMeans}(\mathcal{E}, K_c)$

\For{$j=1$ {\bfseries to} $K_c$}
    \State $\textsc{Quality}(\mathcal{P}_j) \leftarrow \textsc{Mean}\big(\{\textsc{Score}(C): C \in \mathcal{P}_j\}\big)$
    \hfill \textcolor{blue}{$\triangleright$ cluster-level quality}
\EndFor
\State $P^* \leftarrow \arg\max_{j}\textsc{Quality}(\mathcal{P}_j)$

\State \textcolor{blue}{/* \hspace{0.2em} Template Abstraction \hfill */} 
\State $\mu \leftarrow \textsc{Centroid}(P^*)$
\State $\mathcal{X} \leftarrow \textsc{TopNByCosSim}(P^*, \mu, N)$
\hfill \textcolor{blue}{$\triangleright$ pick $N$ exemplars near the centroid}
\State $T^* \leftarrow G.\textsc{AbstractPattern}(\mathcal{X}, \texttt{meta-prompt})$
\hfill \textcolor{blue}{$\triangleright$ synthesize Optimal CoT Template}

\State {\bfseries Return:} $P^*,\ T^*$
\end{algorithmic}
\end{algorithm}

\begin{algorithm}[tb]
\caption{SCoTER: Offline Training and Online Inference (No Online LLM Calls)}
\label{alg:scoter}
\begin{algorithmic}[1]
\State {\bfseries Input :} Training data $\mathcal{D}_{train}$ (splits $\mathcal{D}_{val},\mathcal{D}_{test}$), discovered pattern $P^*$ and template $T^*$,
Teacher LLM $G_T$, Student LLM $G_S$,
Backbone $\textsc{Backbone}(\cdot)$,
sentence encoder $\textsc{SentenceEncoder}(\cdot)$,
item embedding table $\textsc{ItemEmb}(\cdot)$,
fusion modules $\{\textsc{Adapter}, \textsc{PosEnc}, \textsc{CrossAttn}, \textsc{GatedFuse}\}$, weight $\lambda$
\State {\bfseries Output :} Optimized recommendation parameters $\Theta$ and distilled student $G_S$

\State \textcolor{blue}{/* \hspace{0.2em} Step 1: Structured Distillation (Offline) \hfill */} 
\State $\mathcal{D}_{cot} \leftarrow \emptyset$
\For{{\bfseries each} $(S_i, i_{T+1}^{(i)}) \in \mathcal{D}_{train}$}
    \State $C_i \leftarrow G_T.\textsc{Generate}(S_i, T^*)$ \hfill \textcolor{blue}{$\triangleright$ template-consistent CoT}
    \State $\mathcal{D}_{cot} \leftarrow \mathcal{D}_{cot} \cup \{(S_i, C_i)\}$
\EndFor
\State $G_S \leftarrow \textsc{FineTune}(G_S, \mathcal{D}_{cot})$

\State \textcolor{blue}{/* \hspace{0.2em} Step 2: Precompute Reasoning Embeddings (Offline) \hfill */} 
\For{{\bfseries each} $S_i \in \mathcal{D}_{train} \cup \mathcal{D}_{val} \cup \mathcal{D}_{test}$}
    \State $C_i \leftarrow G_S.\textsc{Generate}(S_i, T^*)$
    \State Extract steps $C_i=(c_{i,1}, c_{i,2}, \ldots, c_{i,K})$
    \For{$j=1$ {\bfseries to} $K$}
        \State $\mathbf{e}_{i,j} \leftarrow \textsc{SentenceEncoder}(c_{i,j})$
        \hfill \textcolor{blue}{$\triangleright$ step embedding}
    \EndFor
    \State $\mathbf{H}_i \leftarrow [\mathbf{e}_{i,1}; \mathbf{e}_{i,2}; \ldots; \mathbf{e}_{i,K}]$
    \hfill \textcolor{blue}{$\triangleright$ structured reasoning matrix}
    \State $\textsc{Store}(S_i \mapsto \mathbf{H}_i)$ \hfill \textcolor{blue}{$\triangleright$ offline cache}
\EndFor

\State \textcolor{blue}{/* \hspace{0.2em} Step 3: Train Recommendation Model (Offline) \hfill */} 
\Repeat
    \State Sample mini-batch $\mathcal{B}=\{(S, i_{T+1})\} \sim \mathcal{D}_{train}$
    \State $\mathbf{h}_{\text{seq}} \leftarrow \textsc{Backbone}(S)$ \hfill \textcolor{blue}{$\triangleright$ sequence embeddings}
    \State $\mathbf{H} \leftarrow \textsc{Retrieve}(S)$ \hfill \textcolor{blue}{$\triangleright$ no LLM call}

    \State \textcolor{blue}{/* \hspace{0.2em} Order-Preserving Fusion \hfill */} 
    \State $\mathbf{Z} \leftarrow \textsc{Adapter}(\mathbf{H})$ \hfill \textcolor{blue}{$\triangleright$ projection + normalization}
    \State $\mathbf{Z}^{pos} \leftarrow \mathbf{Z} + \textsc{PosEnc}$ \hfill \textcolor{blue}{$\triangleright$ positional encoding}
    \State $\mathbf{A} \leftarrow \textsc{CrossAttn}(\mathbf{h}_{\text{seq}}, \mathbf{Z}^{pos}, \mathbf{Z}^{pos})$
    \hfill \textcolor{blue}{$\triangleright$ cross-attention}
    \State $\mathbf{E}_{\text{fused}} \leftarrow \textsc{GatedFuse}(\mathbf{h}_{\text{seq}}, \mathbf{A})$
    \hfill \textcolor{blue}{$\triangleright$ gated fusion}

    \State \textcolor{blue}{/* \hspace{0.2em} Joint Objective \hfill */} 
    \State $\mathcal{L}_{rec} \leftarrow -\log p(i_{T+1} \mid \mathbf{E}_{\text{fused}})$
    \State $\mathbf{z}_{K}^{pos} \leftarrow \mathbf{Z}^{pos}[K]$ \hfill \textcolor{blue}{$\triangleright$ final reasoning step}
    \State $\mathcal{L}_{\text{InfoNCE}} \leftarrow \textsc{InfoNCE}(\mathbf{z}_{K}^{pos}, \textsc{ItemEmb}(i_{T+1}))$
    \State $\mathcal{L}_{total} \leftarrow \mathcal{L}_{rec} + \lambda \mathcal{L}_{\text{InfoNCE}}$
    \State Update $\Theta \leftarrow \Theta - \eta \nabla_{\Theta}\mathcal{L}_{total}$
\Until{convergence}

\State \textcolor{blue}{/* \hspace{0.2em} Online Inference (No Online LLM Calls) \hfill */} 
\State {\bfseries Input :} User sequence $S$
\State $\mathbf{h}_{\text{seq}} \leftarrow \textsc{Backbone}(S)$
\State $\mathbf{H} \leftarrow \textsc{Retrieve}(S)$ \hfill \textcolor{blue}{$\triangleright$ cached $\mathbf{H}\in\mathbb{R}^{K\times D}$}
\State $\mathbf{E}_{\text{fused}} \leftarrow \textsc{Fuse}(\mathbf{h}_{\text{seq}}, \mathbf{H})$ \hfill \textcolor{blue}{$\triangleright$ same as training}
\State {\bfseries Return :} $\textsc{RankAllItems}(\mathbf{E}_{\text{fused}})$ \hfill \textcolor{blue}{$\triangleright$ full-catalog ranking}
\end{algorithmic}
\end{algorithm}

\newpage
~\newpage

~\newpage

\section{Ablation Studies and Pattern Analysis on SCoTER-T}~\label{ab:on_tiger}

In this section, we provide a detailed analysis of the automated pattern discovery and architectural components. Note that all experiments in this section are conducted using the SCoTER-T variant (TIGER backbone) on the Beauty dataset.

\begin{figure}[t]
\centering
\includegraphics[width=1.0\linewidth]{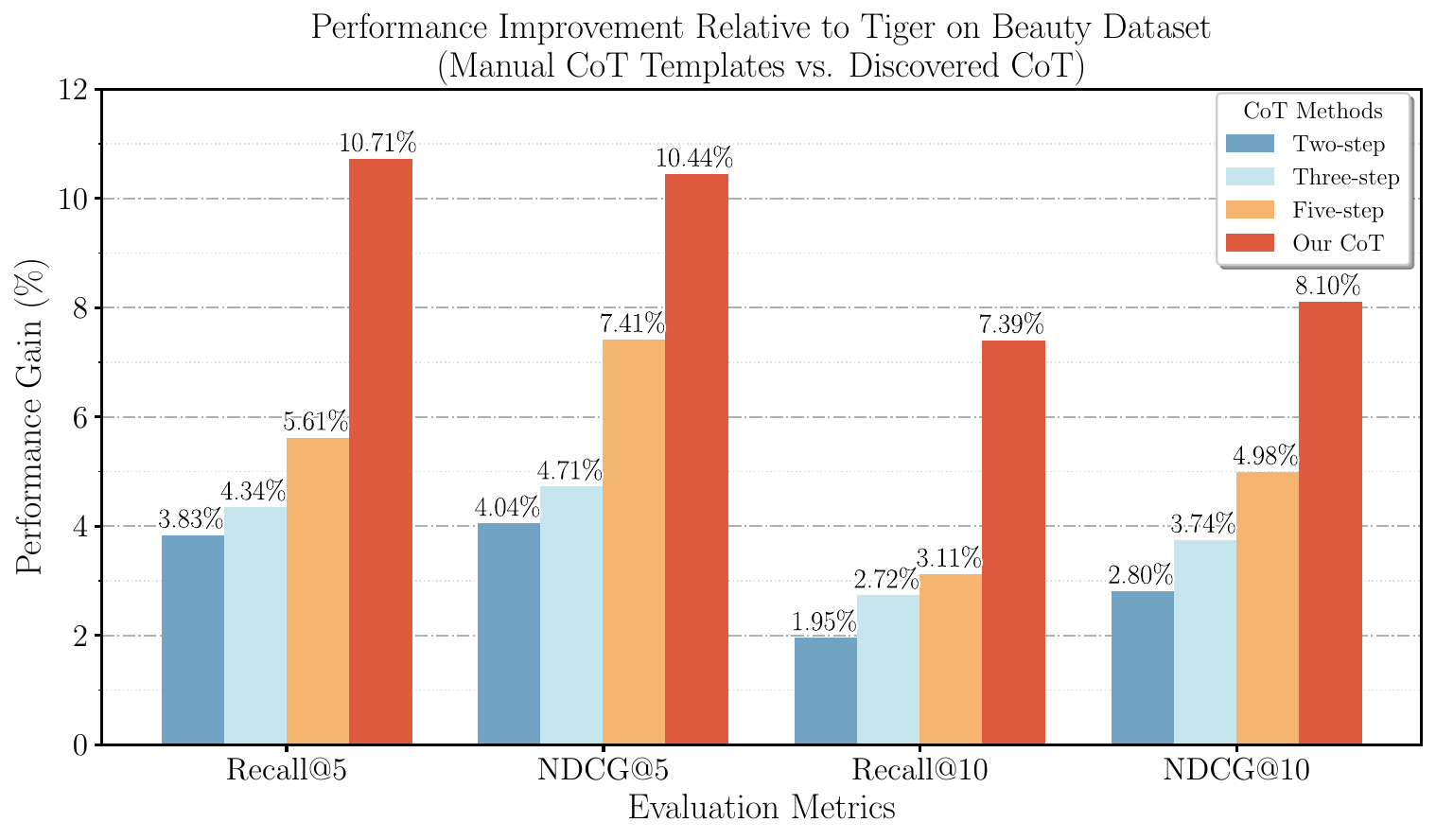}
\caption{Performance Improvement Relative to Tiger on \textbf{Beauty} with backbone integration. 
Manual CoT Templates (Two-step, Three-step, Five-step) are compared with the automatically Discovered CoT.}
\label{fig:beauty_k5k10_gain}
\end{figure}

\subsection{Automated Pattern Discovery (RQ2)}
To rigorously evaluate our automated discovery within the SCoTER-T framework, we compared the GVM-discovered pattern against several manual templates. These manual templates, detailed in Appendix~\ref{sec:appendix_prompts}, represent general-purpose reasoning structures derived from domain knowledge and expert intuition. As shown in Figure~\ref{fig:beauty_k5k10_gain}, when integrated with the TIGER backbone, the GVM-discovered pattern demonstrates a substantial advantage over the manual alternatives across all metrics on the Beauty dataset. Notably, its 10.71\% improvement in Recall@5 nearly doubles the gain of the best-performing manual template, establishing a robust performance gap with gains ranging from 7.39\% (NDCG@5) to 10.71\%.

\begin{table*}[t]
\centering
\caption{Ablation results on Beauty dataset.}
\label{tab:ablation}
\begin{tabular}{lcccc}
\toprule
Variant & Recall@5 &Recall@10 & NDCG@5 & NDCG@10 \\
\midrule
Full model        & \textbf{0.0434 \footnotesize$(-)$}& \textbf{0.0656 \footnotesize$(-)$}  & \textbf{0.0276 \footnotesize$(-)$} &\textbf{0.0347 \footnotesize$(-)$} \\
w/o Position     &0.0424 ($\downarrow 2.30\%$) & 0.0647 ($\downarrow 1.37\%$) & 0.0270 ($\downarrow 2.17\%$) &0.0341 ($\downarrow 1.73\%$) \\
w/o Contrastive   &0.0413 ($\downarrow 4.84\%$)& 0.0639 ($\downarrow 2.59\%$)& 0.0267 ($\downarrow 3.26\%$)&0.0337 ($\downarrow 2.88\%$)\\
w/o Step-wise CoT embedding &0.0407 ($\downarrow 6.22\%$) & 0.0624 ($\downarrow 4.88\%$)& 0.0265 ($\downarrow 3.99\%$)&0.0335 ($\downarrow 3.46\%$) \\
    w/o Student Distill &0.0394 ($\downarrow 9.22\%$) & 0.0609 ($\downarrow 7.16\%$)& 0.0262 ($\downarrow 5.07\%$)&0.0331 ($\downarrow 4.61\%$) \\
Tiger & 0.0392 ($\downarrow 9.68\%$) & 0.0594 ($\downarrow 9.45\%$)& 0.0257 ($\downarrow 6.88\%$)& 0.0321  ($\downarrow 7.49\%$)\\
\bottomrule
\end{tabular}
\end{table*}

This superiority extends beyond integrated settings to standalone LLM generation (Table~\ref{tab:llm_recommender}). On both a fine-tuned Qwen3-8B and the larger DeepSeek-R1, our pattern consistently outperforms the manual alternative. This consistent outperformance across diverse settings and models highlights a fundamental architectural advantage, which stems from the GVM pipeline’s systematic, data-driven approach.

This advantage can be understood by deconstructing the GVM process. Manual templates usually rely on generalized human experience. While providing a reasonable starting point, this generality prevents them from capturing the fine-grained, dynamic signals specific to current user interactions, thus limiting their practical effectiveness. Conversely, our GVM pipeline systematically uncovers superior patterns. The Generate phase explores a vast landscape of potential reasoning patterns directly from the data, moving beyond pre-defined assumptions. Crucially, the Validate phase acts as an empirical filter, scoring each candidate based on its actual recommendation performance, thereby creating a feedback loop that ensures only data-supported reasoning paths survive. Finally, the Mine phase distills the most effective and generalizable logic from this validated set. This systematic discovery process allows us to identify latent, data-specific reasoning structures that are not just theoretically sound, but empirically proven to be more beneficial.

\subsection{Structure-Preserving Integration (RQ3)}~\label{append:inter_ab_on_tiger}
To validate the structure-preserving architecture of SCoTER-T, we conducted a systematic ablation study (Table~\ref{tab:ablation}). The results demonstrate that each component is essential, with their removal causing measurable performance degradation ranging from 2.30\% to 6.22\% in Recall@5.

Step-wise CoT embedding emerges as the most critical component, yielding the largest performance degradation, with the Recall@5 score dropping by 6.22\%. It preserves the progressive refinement inherent in reasoning chains. Each step builds upon previous insights to iteratively narrow the recommendation space. Collapsing this multi-step structure into a single vector discards these intermediate logical dependencies, forcing the model to recommend without the benefit of stepwise deliberation.

\begin{table}[t]
\centering
\caption{LLM-as-recommender performance on Beauty dataset. Models generate recommendations directly from reasoning chains without backbone integration.}

\label{tab:llm_recommender}
\setlength{\tabcolsep}{2.0pt}
\begin{tabular}{lcccc}
\toprule
\multirow{2}{*}{} & \multicolumn{2}{c}{DeepSeek-R1} & \multicolumn{2}{c}{Qwen3-8B (Fine-tuned)} \\
\cmidrule(lr){2-3} \cmidrule(lr){4-5}
& Recall@20 & NDCG@20 & Recall@20 & NDCG@20 \\
\midrule
Two-step & 0.0078 & 0.0041 & 0.0340 & 0.0138 \\
Three-step & 0.0089 & 0.0047 & 0.0344 & 0.0142 \\
Five-step & 0.0098 & 0.0052 & \underline{0.0352} & \underline{0.0145} \\
SCoTER & \textbf{0.0105} & \textbf{0.0056} & \textbf{0.0363} & \textbf{0.0152} \\
\bottomrule
\end{tabular}
\end{table}

Beyond this structural foundation, positional encoding and contrastive learning provide complementary enhancements. First, positional encoding preserves sequential order. Without explicit positional signals, the model struggles to differentiate between an early hypothesis exploration and a final refinement. This ambiguity hinders the application of appropriate attention weights across different reasoning stages, thereby degrading the model's ability to leverage the sequential structure.  Consequently, removing this component leads to a significant 2.30\% drop in Recall@5. Second, contrastive learning aligns this reasoning with recommendation objectives. It provides a crucial supervisory signal that steers the logic beyond mere internal coherence to match user preferences.  Its removal, therefore, causes an even larger degradation, with the Recall@5 score dropping by 4.84\% to 0.0413.

Finally, the study reveals a synergistic effect that amplifies their individual contributions. Removing both positional encoding and contrastive learning simultaneously results in a performance drop greater than the sum of their individual impacts. This indicates a cooperative relationship: positional encoding preserves the sequential logic, while contrastive learning aligns this logic with recommendation objectives.

\subsection{Integration Synergy (RQ4)}
A pivotal insight is revealed when comparing the outcomes of standalone LLM-based recommendations against our fully integrated model. The best direct-generation configuration—a fine-tuned Qwen3-8B using our CoT pattern—achieves a Recall@20 of 0.0363. In contrast, the integrated approach reaches a substantially higher Recall@10 of 0.0656. This gap highlights the fundamental value of fusing complementary information sources.

Our architecture's advantage stems from its ability to synergize two distinct modalities. LLM generation relies on explicit semantic logic but lacks the implicit collaborative signals that are the foundation of modern recommenders, such as latent patterns of item co-occurrence or user taste clusters. To bridge this gap, the recommender backbone provides strong collaborative priors, while the CoT module injects an interpretable reasoning layer. This fusion creates recommendations that are both empirically grounded and logically justified—a capability neither component possesses alone.

Beyond synergy, the results reveal another key insight: task-specific adaptation is more critical than raw model scale. This is demonstrated by the smaller, fine-tuned Qwen3-8B consistently outperforming the much larger DeepSeek-R1. This outcome validates our structured distillation, demonstrating its ability to transfer sophisticated reasoning into an efficient model. Ultimately, this confirms a viable path for integrating LLM reasoning into large-scale, production-ready systems.

\subsection{Robustness Analysis on Sports Dataset}
To validate whether the robustness trends observed on the Beauty dataset are generalizable, we conducted the same subgroup analysis on the Sports dataset. The results are presented in Figure \ref{fig:group_analysis_sports}.

The trends on the Sports dataset align well with our findings on Beauty, often showing even more pronounced improvements. For Item Groups, the improvement on tail items (Group 0) is exceptionally high at 49.84\%. This reinforces the conclusion that our reasoning-enhanced framework significantly boosts the discoverability of long-tail items, likely because the reasoning module can leverage item text descriptions to find matches even when interaction data is scarce.

For User Groups, we again observe universal improvements. User Group 2 exhibits the highest relative gain of 45.50\%, while the cold-start users in Group 0 also see a massive 31.28\% boost. These results confirm that SCoTER provides a robust enhancement mechanism that adapts well to varying degrees of user activity and item popularity across different domains.

\begin{figure}[h]
\centering
\includegraphics[width=1.0\linewidth]{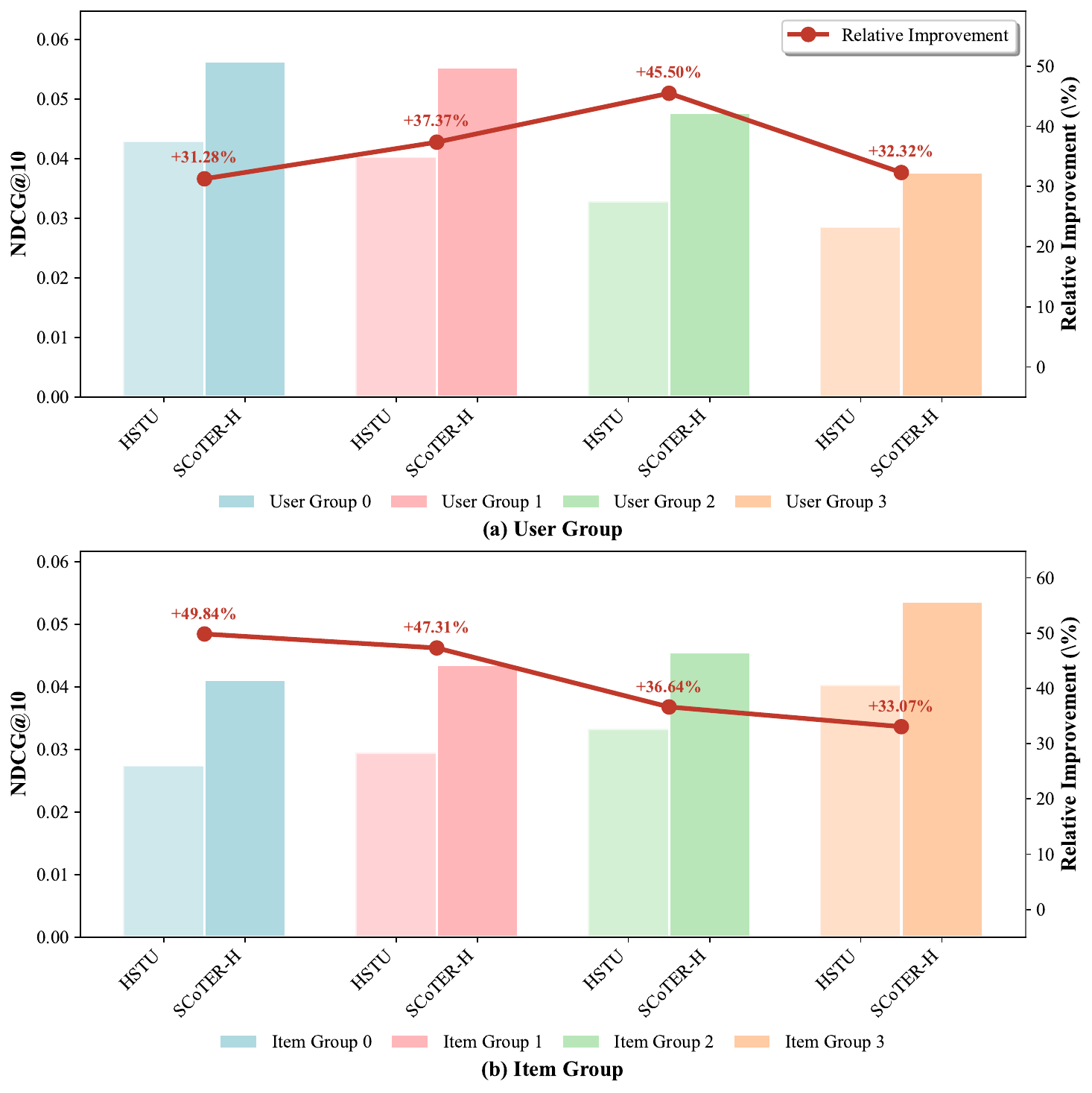}
\caption{Performance comparison across different user and item groups on the Sports dataset. Higher group numbers indicate higher interaction frequency.}
\label{fig:group_analysis_sports}
\end{figure}

\section{Theoretical Analysis}
\label{sec:appendix_theory}

In this section, we provide the formal definitions, objectives, and rigorous proofs supporting the design of our Structure-Preserving Integration module.

\subsection{Formal Definitions}

We first define the key components of our problem setting.

\begin{itemize}
    \item \textbf{Sequential Recommendation:} Given users $\mathcal{U}$ and items $\mathcal{I}$, each user $u \in \mathcal{U}$ has a chronologically ordered interaction history $S = [i_1, \ldots, i_T] \in \mathcal{I}^T$. The goal is to learn a model $q_\theta$ that approximates the ground-truth next-item distribution $p^*(Y| S)$.

    \item \textbf{Reasoning Pattern ($P$):} A pattern $P=(p_1, \ldots, p_k) \in \mathcal{P}$, with a fixed length $k$, is a high-level reasoning template, e.g., $P=$ ("Analyze history" $\rightarrow$ "Identify preferences" $\rightarrow$ "Predict features" $\rightarrow$ "Recommend items").

    \item \textbf{Reasoning Chain ($C$):} For a sequence $S$ and pattern $P$, a reasoning chain $C = (c_1, \ldots, c_k)$ is generated by a pattern-conditioned LLM, denoted $C \sim G_{P}(S)$. Each sentence $c_j$ instantiates the template $p_j$ with user-specific details. The space of all possible chains is denoted as $\mathcal{C}$.

    \item \textbf{Encoders:} We define two types of encoders:
    \begin{itemize}
        \item An encoder $\psi: \mathcal{C} \to \mathbb{R}^{k \times d}$ is \textbf{order-sensitive} if $\psi(C) \neq \psi(C_\pi)$ for some permutation $\pi \neq \text{id}$. It represents the chain as a sequence of $k$ step-embeddings (e.g., via Transformers).
        \item An encoder $\phi: \mathcal{C} \to \mathbb{R}^{d}$ is \textbf{order-agnostic} if $\phi(C) = \phi(C_\pi)$ for all permutation $\pi$. It collapses the sequence of step-embeddings into a single $d$-dimensional vector representation (e.g., via mean pooling).
    \end{itemize}

    \item \textbf{$(\rho,\delta)$-Order Sensitivity:} A task is \emph{$(\rho,\delta)$-order sensitive} if with probability at least $\rho$, for a user sequence $S$, a reasoning chain $C$ can be generated whose predictive distribution changes by at least $\delta$ (in Total Variation distance) under step permutation. Formally, $\Pr(S \in \Omega_\delta) \ge \rho$, where 
    \[
    \Omega_\delta = \{S \mid \exists C \sim G_P(S), \pi \neq \text{id} \text{ s.t. } \text{TV}(q_\theta(\cdot|S,C), q_\theta(\cdot|S,C_\pi)) \ge \delta \}.
    \]
\end{itemize}

\subsection{Optimization Objective and Justification}

\paragraph{Optimization Objective.}
To jointly identify an optimal pattern $P^*$ and train a model $\theta$ that approximates $p^*(Y|S)$, our framework maximizes the expected log-likelihood by marginalizing over chains $C \sim G_{P^{*}}(S)$:
\begin{equation}
\label{eq:objective}
\max_{\theta} \mathbb{E}_{S,Y \sim p^*} \left[ \log \mathbb{E}_{C \sim G_{P^{*}}(S)} \left[ q_\theta(Y|S,C) \right] \right]
\end{equation}
This objective effectively decouples pattern discovery (finding $P^*$) from model training (optimizing $\theta$).

\paragraph{Information-Theoretic Justification.}
Our framework's architecture is motivated by decomposing the predictive value of a reasoning chain, $I(C;Y|S)$. Using an operator $f(C)=P$ to extract the pattern from a chain, this value decomposes as:
\begin{equation}
\label{eq:info_decomp}
I(C;Y| S) \;=\; I(f(C);Y| S) + I(C;Y| S, f(C))
\end{equation}
The first term, $I(P;Y|S)$, motivates our GVM pipeline for \textbf{Pattern Discovery}. The second term, $I(C;Y|S,P)$, quantifies the value of the chain's ordered details, motivating our \textbf{Structure Preservation} architecture.

\subsection{Theorems and Proofs}

We now provide the formal proofs for the advantages of preserving sequential order.

\begin{theorem}[Information-Theoretic Advantage]
\label{theorem:Advantage}
Let $\mathcal{H}_{\text{seq}} = \psi(C)$ and $\mathcal{H}_{\text{bag}} = \phi(C)$ be representations from order-sensitive and order-agnostic encoders, respectively. Since $\mathcal{H}_{\text{bag}}$ can be derived from $\mathcal{H}_{\text{seq}}$, the Data Processing Inequality implies:
\[
I(\mathcal{H}_{\text{seq}}; Y | S)\ \ge\ I(\mathcal{H}_{\text{bag}}; Y | S)
\]
\end{theorem}

\begin{proof}
Since the order-agnostic representation $\mathcal{H}_{\text{bag}}$ is a deterministic function of the order-sensitive representation $\mathcal{H}_{\text{seq}}$ (e.g., via operations like mean pooling or summation that discard positional information), the Markov chain $Y \to \mathcal{H}_{\text{seq}} \to \mathcal{H}_{\text{bag}}$ holds given $S$. The result follows directly from the Data Processing Inequality, which states that no local processing of a variable ($\mathcal{H}_{\text{seq}}$ to $\mathcal{H}_{\text{bag}}$) can increase its mutual information with the target variable $Y$.
\end{proof}

\begin{lemma}[Performance Lower Bound]
\label{lemma:LowerBound}
For any model $q_\theta$, the expected recall is lower-bounded by:
\[
\mathbb{E}[\operatorname{Recall@K}] \ge \mathbb{E}[m_K(S,C)] - \mathbb{E}\left[ \mathrm{TV}\left( p^*_S, q_\theta(\cdot | S,\mathrm{Encoder}(C)) \right) \right]
\]
where $m_K(S,C)$ is the sum of probabilities assigned by the model to the top-$K$ predicted items, and $p^*_S$ denotes the ground-truth distribution $p^*(\cdot|S)$.
\end{lemma}

\begin{proof}
Let $A_q$ be the set of top-$K$ items predicted by $q_\theta$. For a given sequence $S$, the expected Recall@K is the probability that the ground truth item falls within this set:
\[
\mathbb{E}[\text{Recall@K} | S] = \mathbb{E}[\mathbb{I}_{Y \in A_q} | S] = p^*_S(A_q).
\]
By the definition of Total Variation (TV) distance, for any event $A$: $|p^*_S(A) - q_\theta(A|S)| \le \text{TV}(p^*_S, q_\theta)$.
Rearranging this inequality gives a lower bound for the ground truth probability:
\[
p^*_S(A_q) \ge q_\theta(A_q|S) - \text{TV}(p^*_S, q_\theta).
\]
By definition, $q_\theta(A_q|S) = \sum_{y \in A_q} q_\theta(y|S) = m_K$. Substituting this back:
\[
\mathbb{E}[\text{Recall@K} | S] \ge m_K - \text{TV}(p^*_S, q_\theta).
\]
Taking the expectation over all samples $S$ yields the lemma.
\end{proof}

To prove the final advantage, we first introduce a supporting lemma regarding the error inherent in order-agnostic encoding.

\begin{lemma}[Collision Penalty]
For any $S \in \Omega_\delta$, there exists a permutation $\pi \neq \text{id}$ such that for an order-agnostic encoder $\phi$:
\[
\max\{\text{TV}(p^*(\cdot|S,C), q_\phi), \text{TV}(p^*(\cdot|S,C_\pi), q_\phi)\} \ge \frac{\delta}{2}
\]
\end{lemma}
\begin{proof}
Since $S \in \Omega_\delta$, we have $\text{TV}(p^*(\cdot|S,C), p^*(\cdot|S,C_\pi)) \ge \delta$. An order-agnostic encoder produces the exact same representation for both $C$ and $C_\pi$, thus $q_\phi(\cdot|\phi(C)) = q_\phi(\cdot|\phi(C_\pi))$. By the triangle inequality:
\begin{align*}
\delta &\le \text{TV}(p^*(\cdot|S,C), p^*(\cdot|S,C_\pi)) \\
&\le \text{TV}(p^*(\cdot|S,C), q_\phi) + \text{TV}(q_\phi, p^*(\cdot|S,C_\pi))
\end{align*}
If both terms on the right were less than $\delta/2$, their sum would be less than $\delta$, leading to a contradiction. Thus, at least one term must be $\ge \delta/2$.
\end{proof}

\begin{theorem}[Order-Aware Performance Advantage]
\label{theorem:PerformanceAdvantage}
For a $(\rho,\delta)$-order sensitive task, an order-sensitive encoder $\psi$ achieves a performance advantage over an order-agnostic encoder $\phi$:
\begin{align*}
\mathbb{E}_\psi[\text{Recall@K}] - \mathbb{E}_\phi[\text{Recall@K}]
&\geq \big(\mathbb{E}[m_K]_\psi - \mathbb{E}[m_K]_\phi\big) \\
&\quad + \frac{\rho\delta}{2}
      - \mathbb{E}[\text{TV}(p^*_S, q_\psi)]
\end{align*}
\end{theorem}

\begin{proof}
First, we apply Lemma~\ref{lemma:LowerBound} to both encoders:
\begin{align*}
\mathbb{E}_\psi[\text{Recall@K}] &\ge \mathbb{E}[m_K]_\psi - \mathbb{E}[\text{TV}(p^*_S, q_\psi)] \\
\mathbb{E}_\phi[\text{Recall@K}] &\ge \mathbb{E}[m_K]_\phi - \mathbb{E}[\text{TV}(p^*_S, q_\phi)]
\end{align*}
Subtracting the second inequality from the first (note: to bound the difference $\Delta$, we consider the worst-case lower bound for $\psi$ and best case for $\phi$, but structurally we compare the lower bounds directly to show potential):
\[
\Delta \ge (\mathbb{E}[m_K]_\psi - \mathbb{E}[m_K]_\phi) + (\mathbb{E}[\text{TV}(p^*_S, q_\phi)] - \mathbb{E}[\text{TV}(p^*_S, q_\psi)])
\]
Using the Collision Penalty Lemma, for the order-agnostic encoder $\phi$, the expected error on the sensitive set $\Omega_\delta$ is at least $\delta/2$. Thus, the total expected error is bounded by:
\[
\mathbb{E}[\text{TV}(p^*_S, q_\phi)] \ge \rho \cdot \frac{\delta}{2} + (1-\rho) \cdot 0 = \frac{\rho\delta}{2}.
\]
Substituting this back into the difference equation yields the final result.
\end{proof}

\section{CoT Templates} \label{sec:appendix_prompts}

This section details the manual Chain-of-Thought (CoT) templates used for comparison against our GVM-discovered pattern. These templates represent heuristic-based reasoning structures of increasing complexity. For the experiments, the following strings were used to guide the LLM's reasoning process inside the `<cot$\_$path>` tag.

\textbf{Two-step Template from~\cite{yi2025recgpt}} \\
A direct, two-stage reasoning path focusing on mining interests and then recommending.\\
\textit{\textless cot\_path\textgreater User Interest Mining -> Item Tag Prediction \& Recommendation \textless/cot\_path\textgreater}

\begin{figure*}[htbp]
\newlength{\oldcolumnsep}
\setlength{\oldcolumnsep}{\columnsep}
\setlength{\columnsep}{0pt}

\noindent 
\begin{minipage}{1\textwidth}
\begin{tcolorbox}[colback=white,colframe=white!40!black,
title=\textbf{Prompt: Inferring Recommendation Paths from User Behavior},
    left=2mm,   
    right=2mm,  
    top=2mm,    
    bottom=2mm  
]

\textbf{You are an expert in recommendation algorithms.} Based on the user's historical behavior data, please infer a reasonable analysis and recommendation path, and predict 20 different product features that the user may be interested in the future.

\vspace{1mm}
\textbf{Requirements:}
\begin{enumerate}
    \item First output \texttt{<cot\_path>}, with a reasoning chain of at least two steps, each step $\leq$8 words, connected with "->".
    \item \texttt{<reason>} must be placed after \texttt{<cot\_path>}, and the content must strictly follow each step in \texttt{<cot\_path>} to expand the reasoning one by one.
    \item Finally output 20 recommended product feature descriptions.
\end{enumerate}

\vspace{1mm}
\textbf{Output format requirements:}

\texttt{<cot\_path>} \\
\textit{[Summarize the core reasoning path here: Step 1 -> Step 2-> \textbf{$\cdots$} -> Item recommendation]} \\
\texttt{</cot\_path>}

\vspace{2mm} 

\texttt{<reason>} \\
\textit{[Explain the reasoning logic step by step according to the steps in cot\_path]} \\
\texttt{</reason>}

\vspace{2mm} 

\texttt{<recommendations>} \\
\texttt{<item>}Detailed feature description of product 1, including category, brand, function, characteristics, etc.\texttt{</item>} \\
\centerline{\textbf{$\cdots$}}
\texttt{<item>}Detailed feature description of product 20, including category, brand, function, characteristics, etc.\texttt{</item>} \\
\texttt{</recommendations>}

\vspace{0.5mm}
\end{tcolorbox}
\end{minipage}

\caption{The structure of the prompt used to guide the LLMs for generating candidate CoT paths and recommended item features. The prompt requires a structured output including a concise reasoning path (\texttt{<cot\_path>}), detailed reasoning (\texttt{<reason>}), and 20 specific product feature recommendations (\texttt{<recommendations>}).}
\label{fig:prompt_structure}

\end{figure*}

\textbf{Three-step Template} \\
This template adds an explicit intermediate step for summarizing a user profile before interest extraction.\\
\textit{\textless cot\_path \textgreater User Profile Summary -> User Interest Extraction -> Item Recommendation \textless/cot\_path\textgreater}

\textbf{Five-step Template} \\
A more granular template that breaks down the analysis into multiple distinct phases, from data analysis to feature prediction before the final recommendation.\\
\textit{\textless cot\_path \textgreater Behavioral Data Analysis -> Interest Pattern Recognition -> Preference Trend Analysis -> Predictive Feature Generation -> Targeted Item Recommendation \textless/cot\_path\textgreater}

\paragraph{Mined Pattern (GVM-Discovered)} 
Unlike manual templates based on general heuristics, GVM discovers patterns that capture domain-specific logic (e.g., separating stable vs. recent interests in Beauty, or identifying complementary gear in Sports).

\noindent\textit{\textbf{Beauty (Universal Style):}} \\
\textit{\textless cot\_path \textgreater Analyze sequential purchase history -> Distinguish stable vs. recent interests -> Expand via shared item attributes -> Predict high-confidence features -> Item Recommendation \textless/cot\_path\textgreater}

\noindent\textit{\textbf{Instruments:}} \\
\textit{\textless cot\_path \textgreater Analyze the purchase pattern -> Music interest inference -> Instrument focus identification -> Accessory expansion -> Feature prediction \textless/cot\_path\textgreater}

\noindent\textit{\textbf{Sports:}} \\
\textit{\textless cot\_path \textgreater Purchase pattern recognition -> Outdoor/Athletic interest mapping -> Complementary gear expansion-> Item Recommendation \textless/cot\_path\textgreater}

\noindent\textit{\textbf{Yelp:}} \\
\textit{\textless cot\_path \textgreater Purchase analysis -> Cuisine/Location interest mapping -> Category expansion -> Feature prediction \textless/cot\_path\textgreater}

\begin{figure*}[htbp]
\setlength{\oldcolumnsep}{\columnsep}
\setlength{\columnsep}{0pt}

\noindent
\begin{minipage}{1\textwidth}
\begin{tcolorbox}[colback=white,colframe=white!40!black, 
title=\textbf{Prompt: Mining Optimal Patterns from CoT Clusters},
    left=2mm,
    right=2mm,
    top=2mm,
    bottom=2mm
]

\textbf{You are an expert in recommendation system reasoning patterns.} Analyze the following CoT cluster and propose better reasoning paths.

\textbf{Cluster ID: x}
\begin{itemize}
    \item \textbf{Performance}: Rank 1 in K clusters $|$ \{R@20 value\} Recall@20 
    \item \textbf{Representative CoT Paths}:
    \begin{enumerate}
        \item \texttt{<cot\_path >} $\cdots$ \texttt{</cot\_path>}
        \item[] \hspace{1cm} \textbf{$\cdots$} 
        \item \texttt{<cot\_path >} $\cdots$ \texttt{</cot\_path>}
    \end{enumerate}
\end{itemize}

\noindent\hrulefill 

\textbf{Analysis \& Output:}

\textbf{\#\#\# 1. Pattern Summary}
\begin{itemize}
    \item \textbf{Core Logic}: \textit{[One sentence describing the main reasoning approach]}
    \item \textbf{Key Steps}: \textit{[The most important characteristics of this pattern]}
\end{itemize}

\textbf{\#\#\# 2. Performance Assessment}
\begin{itemize}
    \item \textbf{Why it works}: \textit{[What makes this pattern achieve its current hit rate]}
    \item \textbf{Limitations}: \textit{[Main weaknesses preventing higher performance]}
\end{itemize}

\textbf{\#\#\# 3. Optimized Paths}
\begin{itemize}
    \item \textbf{Best Path (each step $\leq$ 8 words)}: \textit{[Step1 -> Step2 -> ... -> Stepk -> Item Recommendation]}
    \item \textbf{Improvement Logic}: \textit{[How these address the identified limitations and why they should perform better]}
\end{itemize}

\vspace{1mm}
\textbf{Requirements:} Keep the CoT path \textbf{clear and not overly complex}

\end{tcolorbox}
\end{minipage}

\caption{The prompt structure used in the Mine phase. The LLM acts as an expert analyst to abstract a single, optimal reasoning template from a cluster of high-performing candidate chains.}
\label{fig:mining_prompt}

\end{figure*}

\begin{figure*}[htbp]
\centering
\begin{tcolorbox}[
    colback=white,
    colframe=white!40!black, 
    boxrule=0.8pt,
    arc=2pt,
    title=\textbf{Example},
    coltitle=white,
    fonttitle=\bfseries,
    left=2mm, right=2mm, top=2mm, bottom=2mm
]
\small

\textbf{[Instruction Template]} \\
You are an expert recommendation system. Your task is to analyze user purchase patterns and predict their next likely purchase based on their sequential purchase history.

\vspace{1mm}
\textbf{Reasoning Pattern (High-level Steps)}: 
Analyze sequential purchase history $\to$ Distinguish stable vs. recent interests $\to$ Expand via shared item attributes $\to$ Predict high-confidence features $\to$ Item Recommendation

\vspace{1mm}
\textbf{Required Output Format}:
Start with \textbf{Overall Analysis}. Then provide detailed analysis for each step in the Reasoning Pattern. End with \textbf{Summary}.

\vspace{2mm}
\textbf{[Input Data]} \\
Analyze the following user purchase history and provide a comprehensive \textbf{structured rationale} following the exact format:

\vspace{1mm}
\textit{Purchase History in Chronological Order:}
\begin{itemize}
    \item \textbf{\#\# Item 1}: \textit{OPI Nail Lacquer, Simmer and Shimmer}. Desc: OPI Burlesque Colors.
    \item \textbf{\#\# Item 2}: \textit{OPI Red Shatter Crackle Nail Polish}. Desc: Red Shatter Nail Polish.
    \item \textbf{\#\# Item 3}: \textit{SKIN79 Prestige Beblesh Balm BB Cream}. Desc: Whitening, wrinkles improvement... Diamond Powder...
    \item \textbf{\#\# Item 4}: \textit{WAWO 15 Color Professional Camouflage Concealer}. Desc: Camouflage for blemishes, scars... perfect for under-eye area...
\end{itemize}

\vspace{2mm}
\hrule 
\vspace{2mm}

\textbf{[Model Output: Structured Rationale]}

\vspace{1mm}
\textbf{Overall Analysis}: The user demonstrates a clear focus on detailed personal grooming. Their behavior suggests a preference for "professional-grade" home beauty solutions, transitioning from expressive decoration to functional care.

\vspace{1mm}
\textbf{[Analyze sequential purchase history]}: The sequence begins with vibrant hand aesthetics (Items 1-2), prioritizing visual impact. It then shifts chronologically towards corrective facial skincare (Items 3-4), indicating a move from purely decorative cosmetics to functional, skin-improving formulations.

\vspace{1mm}
\textbf{[Distinguish stable vs. recent interests]}:
\textbf{Stable:} A consistent demand for "Salon-quality/Professional tools" (e.g., OPI brand, 15-color palette).
\textbf{Recent:} The immediate focus has pivoted to "Skin Health" and "Restoration" (whitening, moisturizing), showing a specific intent to improve skin condition rather than just painting over it.

\vspace{1mm}
\textbf{[Expand via shared item attributes]}: We synthesize the \textit{Location} from the stable history (Hands) with the \textit{Function} from the recent history (Therapeutic/Skin Softening). The expansion leads to "Restorative Hand Treatments."

\vspace{1mm}
\textbf{[Predict high-confidence features]}: The next item should likely offer: "Thermal/Heating" (for deep treatment), "Moisturizing" (aligning with skincare logic), and "Spa-like efficacy."

\vspace{1mm}
\textbf{[Item Recommendation]}: \textbf{A professional-grade thermal paraffin wax bath or home spa kit}.
\textit{Rationale}: This product category perfectly synthesizes the user's profile. It addresses the "Hand Care" interest established by the OPI purchases but elevates it to the "Treatment/Care" level seen in the recent skincare purchases.

\end{tcolorbox}
\caption{A complete example of data synthesis. The upper part shows the pattern-injected instruction and user history provided to the Teacher LLM. The lower part shows the structured rationale generated by the LLM.}
\label{fig:full_synthesis_example}
\end{figure*}
\end{document}